%% file: main_cs_v2.tex
\documentclass[twocolumn]{IEEEtran}
 
\usepackage{latexsym}
\usepackage{amssymb}
\usepackage{bm}
\usepackage{stmaryrd}
\usepackage[dvips]{graphics}
\usepackage{graphicx}
\usepackage{psfrag}
\usepackage{amsfonts,amsmath,amssymb}
\usepackage{algorithm}
\usepackage{algorithmic,amsthm}
\usepackage{dsfont,color,subfigure,setspace}
\usepackage{cite}

\DeclareMathOperator*{\argmin}{arg\,min}

\input{defns.tex}

\newcommand{\ex}{{\rm e}}
\newcommand{\y}{\yv }

\newcommand{\xvt}{\tilde{\bf x} }
\newcommand{\xvh}{\hat{\bf x} }

\newcommand{\cv}{{\bf c}}

\newcommand{\s}{\sigma}

\newtheorem{definition}{Definition}

\newtheorem{theorem}{Theorem}
\newtheorem{lemma}{Lemma}

\newtheorem{corollary}{Corollary}
\newtheorem{remark}{Remark}
\newtheorem{proposition}{Proposition}
\newtheorem{example}{Example}

\begin{document}

\title{From Compression to Compressed Sensing\footnote{This paper was presented in part at  Information Theory and Applications Workshop, San Diego, CA 2013, at IEEE International Symposium on Information Theory,  Istanbul, Turkey, 2013,   and at Signal Processing with Adaptive Sparse Structured Representations, EPFL, Lausanne 2013.}}

\author{Shirin Jalali and Arian Maleki}
\date{} 

\maketitle

\begin{abstract}
Can compression algorithms be employed for recovering signals from their underdetermined set of linear measurements? Addressing this question is the first step towards  applying compression algorithms for compressed sensing (CS). In this paper, we consider  a family of compression algorithms $\Cc_r$, parametrized by rate $r$, for a compact class of signals $\Qc \subset \mathds{R}^n$. The set of natural images and JPEG at different rates  are examples of $\Qc$ and $\Cc_r$, respectively.     
We establish a connection  between the rate-distortion performance of $\Cc_r$, and the number of linear measurements required for successful recovery in CS. We then propose compressible signal pursuit (CSP) algorithm and prove that, with high probability, it  accurately and robustly  recovers   signals from an underdetermined set of linear measurements.  We also explore the performance of CSP in the recovery of infinite dimensional signals.  
\end{abstract}

\section{Introduction}
The field of compressed sensing (CS) was established on a keen observation that if a signal is sparse in a certain basis it can be recovered from far fewer random linear measurements than its ambient dimension \cite{Donoho1, CaRoTa06}. In the last decade, CS recovery algorithms have evolved to capture more complicated signal structures  such as group sparsity, low-rankness \cite{BakinThesis, eldar2010block, YuLi06, ji2009multi, MaAnYaBa11, stojnic2009block, stojnic2009reconstruction, stojnic2010ell, MeVaBu08, ChRePaWi10, RichModelBasedCS, ReFaPa10, VeMaBl02, som2012compressive, CaLiMaWr09, ChSaPaWi11, duarte2011performance}, and other broader notions of ``structure'' \cite{DoKaMe06, BaDu12, jalali2011minimum, JaMaBa12}.
In this paper, we consider a different type of structure based on compression algorithms. Suppose that a class of signals can be ``efficiently'' compressed by a compression algorithm. Intuitively speaking,  such classes of signals have a certain ``structure'' that enables the compression algorithm to represent them with fewer bits. These structures are often much more complicated than sparsity, and employing them in CS can potentially  reduce the number of measurements required for signal recovery. 
 
In this paper, we aim to address the following problem. Is it possible to employ compression algorithms in  CS and design compression-based CS  algorithms that  recover  signals  either exactly or with ``small error'',  from their under-determined  set of linear measurements? As we will prove in this paper, the answer to this question is affirmative. 
We propose a CS recovery algorithm based on  exhaustive search over the set of ``compressible'' signals, that, under certain condition on the rate-distortion performance of the code,  recovers signals from fewer measurements   than their ambient dimensions.  This result provides the first theoretical basis for using generic compression algorithms in CS.

The theoretical framework we develop shows a connection between the problem of compressed sensing and relevant problems in the filed of embedology \cite{SaYoCa91, Mane81, HuKa99, BeEdFoNi93}. This connection has also been  explored in \cite{WuVe10}. As an application of this connection, we will derive several fundamental results of embedology as corollaries of our main results. 

We also extend our results   to analog signals. Such extensions are important for many applications including spectrum sensing \cite{TrLaDuRoBa10, MiElDoSh11, MiEl09}. Our generalization employs a new measurement technique, that is based on the projection of the signal onto several independent Wiener processes, and the CSP algorithm. We show that these two ingredients enable us to utilize   most of our proof techniques in the problem of analog-CS as well.

The organization of the paper is as follows. Section \ref{sec:background} reviews the main concepts used in this paper. Section \ref{sec:problemstate} formally states the problem addressed in the paper and introduces our compressible signal pursuit algorithm. Section \ref{sec:lower} establishes a lower bound for the number of measurements any recovery method (based on compression algorithm) requires for accurate recovery.  Section \ref{sec:contributions} and Section \ref{sec:unifCS} summarize our main contributions. Section \ref{sec:extension} extends our results to the class of analog signals. Section \ref{sec:related-work} reviews the related  work in the literature. Section \ref{sec:proofs} includes the proofs of our main theorems. Finally, Section \ref{sec:conclusion} concludes the paper.


\section{Background}\label{sec:background}


\subsection{Notation}
Boldfaced letters such as $\xv$ and $\Xv$ represent vectors. Calligraphic letters denote sets and operators; the distinction will be clear from the context. Given a finite set $\Ac$, $|\Ac|$ denotes its size.  The $\ell_p$-norm of  $\xv \in \mathds{R}^n$ is defined as $\|\xv\|_p \triangleq (\sum_{i=1}^n |x_i|^p)^{1/p}$. The $\ell_0$-norm is also defined as $\|\xv\|_0 \triangleq |\{i \ : \ x_i \neq 0 \}|$. Note that for $p<1$, $\|\cdot \|_p$ is a semi-norm since it does not satisfy the triangle inequality. Throughout the paper $\log$ denotes logarithm in base ${\rm e}$, and logarithm in base 2 is denoted explicitly as $\log_2$.


\subsection{Compression}
Let $\Qc$ denote a compact subset of $\mathds{R}^n$.\footnote{We extend our results to infinite dimensional spaces in Section \ref{sec:extension}.} Consider a compression algorithm for $\Qc$  described by encoder and decoder mappings $(\Ec,\Dc)$. Encoder \[
\Ec: \Qc \to \{1,2, \ldots, 2^{r}\},
\] 
maps  signal $\xv\in\Qc$ to codeword $\Ec(\xv)$.  Decoder  
\[
\Dc:\{1,2, \ldots, 2^r\} \to \hat{\Qc},
\] 
maps the coded signal  $\Ec(\xv)$ back to the reconstruction domain $\hat{\Qc}\subset\mathds{R}^n$. Let $\hat{\xv}\triangleq \Dc(\Ec(\xv))$ denote the reconstruction of signal $\xv\in\Qc$. Let $\Cc$ denote the {\em codebook} of code $(\Ec,\Dc)$, \ie 
\[
\Cc\triangleq \{ \Dc(\Ec(\xv)): \xv \in \Qc\}.
\]
Clearly, $|\Cc|\leq 2^{r}$.
The performance of the described coding scheme is measured in terms of its rate $r$ and its induced  distortion $\d$ defined as
\begin{align}
\d \triangleq \sup_{\xv \in \Qc} \| \xv -\Dc(\Ec(\xv)) \|_2.\label{eq:D}
\end{align}

Consider \emph{a family of compression algorithms} $\{(\Ec_r,\Dc_r): r>0\}$ that are parametrized by  rate $r$. The distortion defined by \eqref{eq:D},  corresponding to  code $(\Ec_r,\Dc_r)$, is denoted by $\d(r)$. Furthermore,  for this family of compression algorithms, we define 
\[
r(\d) \triangleq\inf\{r: \d(r)\leq \d\}.
\]

Note that  $\d(r)$ and similarly  $r(\d)$ are functions  of \emph{both} the family of  compression algorithms and  set $\Qc$, and are different from the Shannon's rate distortion function that characterizes the fundamental limits of compressing a \emph{stochastic} process. Furthermore,  as discussed in Section \ref{sec:kol-d-r}, $r(\d)$ is also different from Kolmogorov's $\e$-entropy. 

The following example clarifies the concepts we have introduced so far. Let $\Qc$ denote the class of all natural images of a certain size and consider the JPEG compression algorithm. Nearly all software implementations of JPEG (including the one in Matlab) provide a parameter that determines the  tradeoff between the size of the compressed file (rate in our terminology) and the quality of the image (distortion in our terminology). Denote this parameter with $\tau$. Fig.~\ref{fig:fig1} shows the distortion-rate performance of JPEG for three different images as $\tau$ varies. As shown in this figure, the performance of the algorithm depends on the image. To characterize $\d(r)$ (defined in \eqref{eq:D}),  we consider \emph{all} natural images and let $\d(r)$  denote the supremum of all achieved distortions  at rate $r$. For instance, according to this definition, the  distortion-rate function of the JPEG algorithm over the class of three images in Fig.~1, is given by the red curve.  Note that characterizing the rate-distortion performance of JPEG or any other heuristic compression algorithm on the class of natural images is computationally prohibitive. However, it is possible to obtain a good lower bound by considering large libraries of natural images.

  \begin{figure}
  \begin{center}
  \includegraphics[width = 7cm]{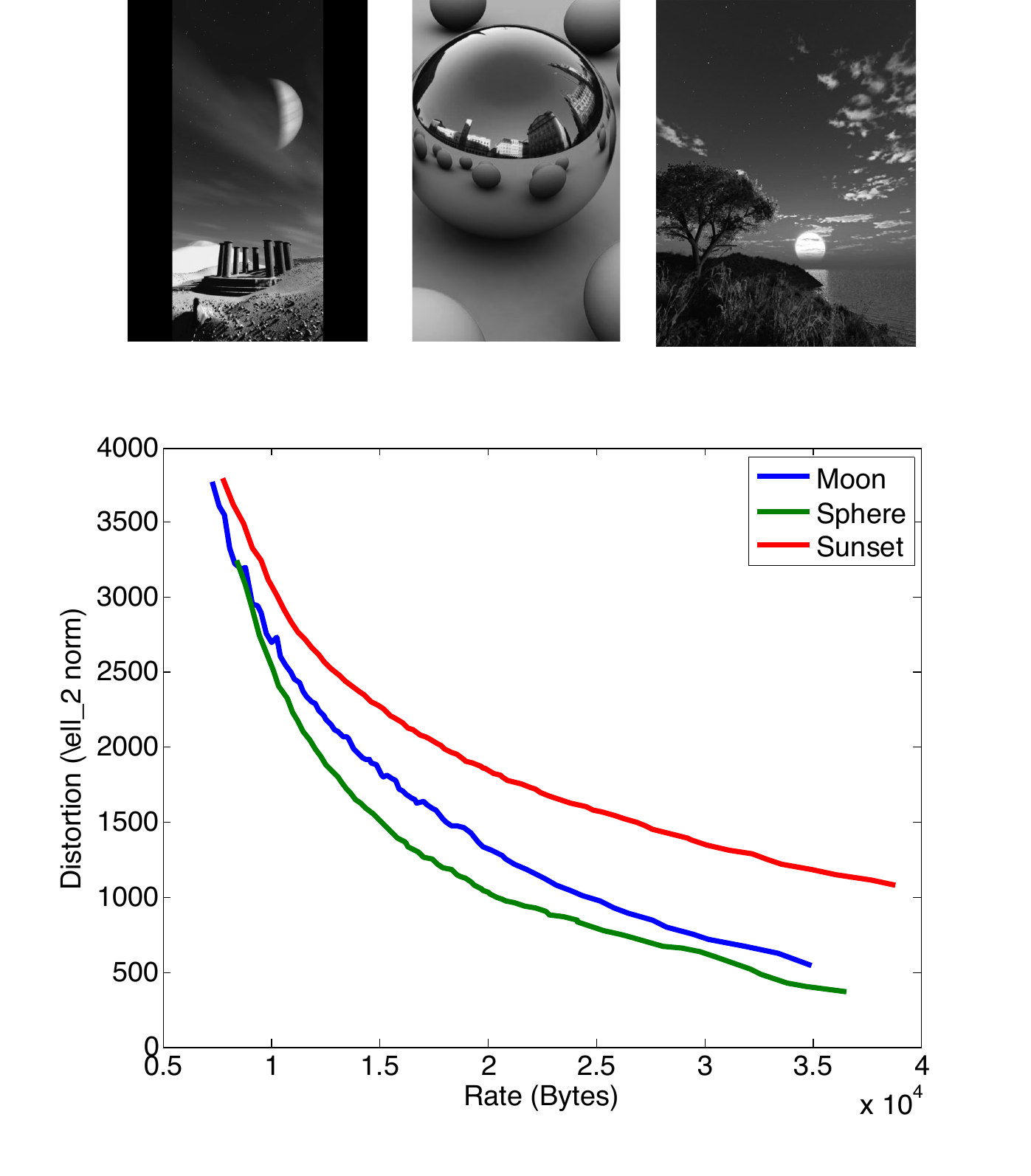}
  \caption{Rate-distrotion performance of JPEG compression for three different images (shown on the top) as $\tau$ varies. }
  \label{fig:fig1}
  \end{center}
\end{figure}

In the remainder   of this section we consider several simpler and theoretically more popular classes of signals to illustrate the concepts. Let $\Bc_p^n(\rho) \triangleq \{\xv \in \mathds{R}^n: \ \|\xv\|_p \leq \rho \}$ represent a ball of radius $\rho$ in $\mathds{R}^n$. Also, let $\Gamma_k^n$ denote the set of all $k$-sparse signals in $\mathds{R}^n$, i.e.,
\begin{equation}\label{eq:ksparse}
\Gamma_k^n \triangleq \{ \xv \in \mathds{R}^n  \ : \ \|\xv\|_0 \leq k\}.
\end{equation}

\begin{example}\label{example:1}
There exists a family of compression algorithms  for $\Bc_2^n(\rho)$ that achieves 
\[
r(\d) \leq \frac{1}{2} n \log_2 n + n \log_2 \left(\frac{\rho}{\d} \right) +cn,
\]
for $\d \leq \rho \sqrt{n}$. Here, $c$ is a constant less than $3$.
\end{example}

\begin{example}\label{example:2}
There exists a family of compression algorithms for $\Bc_2^n(\rho)\cap \Gamma_k^n$ that achieves
\[
r(\d) \leq \log_2{n \choose k} + k \log_2 \left(\frac{\sqrt{k} \rho}{\d} \right) + ck,
\]
where $c$ is a constant less than 3.
\end{example}

These are classic examples in the literature. However, we review  their proofs in Section \ref{proof:examples} to clarify the concepts introduced here. Note that since $\Bc_2^n(\rho)\cap \Gamma_k^n \subset \Bc_2^n(\rho)$, we expect to compress $\Bc_2^n(\rho)\cap \Gamma_k^n$ more efficiently. This is specially clear as $\d \downarrow 0 $.

\subsection{Connection with Kolmogorov's $\epsilon$-entropy}\label{sec:kol-d-r}
The $\epsilon$-entropy of a compact set $\mathcal{Q}$ is defined as 
\[
H_{\epsilon}(\mathcal{Q}) = \log_2 N_{\epsilon}(\mathcal{Q}),
\] 
where $N_\epsilon(\mathcal{Q})$ is the minimum number of elements in an $\epsilon$-covering of $\Qc$ \cite{Koleps59, Clements63}. According to the definition, $H_\epsilon$ provides a lower bound on the rate distortion of any family of compression algorithms.  In other words, if $r(\d)$ is the rate distortion function of a family of compression algorithms on $\mathcal{Q}$, then
\begin{equation}\label{eq:lowerrd}
r(\d) \geq H_{\d}(\Qc).
\end{equation}
As the problem of finding the optimal encoder (the one that achieves $H_{\d}(\Qc)$) is complicated for many real-world signals, heuristic algorithms are being used with rate-distortion performances that might be much worse than $H_{\d}(\Qc)$. Therefore, in this paper we consider a generic compression algorithm that is not necessarily optimal.  We then discuss the implications of our results for $H_{\epsilon}$ and the connection between our work and embedology in Section \ref{sec:relworkepsilon}.\footnote{We are thankful to an anonymous reviewer who pointed out this connection.}


\section{Problem statement}\label{sec:problemstate}

Consider the problem of recovering  ``structured'' signal $\xv\in\mathds{R}^n$ from its undersampled set of linear measurements $\yv = A\xv$,  where $\mathbf{y} \in  \mathds{R}^d$, and $A\in\mathds{R}^{d\times n}$ ($d<n$) denotes the measurement matrix.  For various types of structure such as sparsity, it is well-known that $\xv$ may be recovered from  measurements $\yv$ even with $d<n$. In this paper we explore a more elaborate type of structure based on {\rm compressibility}. 

Instead of being structured as sparse, smooth, etc., suppose that the signal belongs to a  compact set $\Qc\subset \mathds{R}^n$ and there exists a family of compression algorithms $\{(\Ec_r,\Dc_r): r>0\}$ with rate-distortion function $r(\d)$ for signals in $\Qc$. For instance, we can consider the JPEG compression algorithm \cite{JPEGbook} at different rates for the class of images. This family of  compression algorithms might be exploiting the sparsity of the signal in a certain domain or any other type of structure. The actual mechanism by which  the algorithm is  compressing the signals in $\Qc$  is not important for the purpose of this paper. Instead, we are interested in recovering  vector $\xv\in\Qc$ from an undersampled set of linear equations $\yv=A\xv$ by employing the compression algorithms $\{(\Ec_r,\Dc_r): r>0\}$.

 Toward this goal, we follow Occam's principle; Among all the signals that satisfy $\yv = A \xv$, we search for the one that can be compressed well by our compression scheme. More formally, given compression algorithm $(\Ec_r,\Dc_r)$ on $\Qc\subset\mathds{R}^n$ with codebook $\Cc_r$, for recovering $\xv_o\in\Qc$ from its measurements $\yv_o=A\xv_o$, consider  {\em compressible signal pursuit} (CSP) algorithm  defined as
\begin{align}
\hat{\xv}_o=\argmin_{\cv \in \Cc_r}  \|\yv_o- A\cv \|_2^2.\label{eq:CSP}
\end{align}
In case we have access to a family of compression schemes, e.g. JPEG,  the rate $r$ can be considered as a free parameter that can be tuned. The optimal value of this parameter depends on several aspects of the CS system that will become clear later in this section. 

So far we have ignored two important aspects of practical algorithms, which are important in evaluating the performance of CSP:
\begin{itemize}
\item[(i)] Robustness to the measurement noise: The assumption of noiseless measurements, i.e., observing  $\yv_o = A \xv_o$ with no noise, is quite strong for many applications.  A more realistic assumption is to consider $\yv_o = A \xv_o+ \zv$, where $\zv$ denotes the measurement noise. In such settings, we still require to have an ``accurate'' estimate of $\xv_o$. To obtain such an estimate, we can still employ the CSP algorithm. We will show in Sections \ref{sec:measnoisei} and \ref{ssec:noise:ind} that the CSP algorithm is robust to noise, and if the noise is small enough, CSP can still provide an accurate estimate of $\xv_o$. Note that the form of the CSP algorithm is the same for both noiseless and noisy measurements. However, if we have a family of compression algorithms (with $r$ as a free parameter) the optimal value of $r$ depends on the power of the noise. (See Sections \ref{sec:measnoisei} and \ref{ssec:noise:ind} for more details.)

\item[(ii)] Computational complexity: CSP is based on an exhaustive search and hence is computationally very demanding. Practical implementations or approximations  are left for future research.  
\end{itemize}

Before  analyzing the performance of CSP, it is important to determine the conditions under which it is possible to obtain an accurate estimate of $\xv_o$ from $\yv_o=A\xv_o$ by employing a family of compression algorithms. The next section investigates this problem.


\section{CS-applicability}\label{sec:lower}
In this section we address the following questions: Does  existence of a compression algorithm for set $\Qc$ is  equivalent to  existence of a CS-recovery method that can recover  $\xv_o \in \Qc$ from $d < n$ measurements. Under what conditions the existence of a family of compression algorithms leads to a successful recovery algorithm from undersampled set of linear measurements?
Since for every compact set, there exist families of compression algorithms, it seems that the answer to the first question ought to  be negative.  The following lemma confirms this intuition.

\begin{lemma}\label{ex:nocs}
Let $\Qc = \Bc_2^n(1)$. If the number of linear measurements $d$ is less than the ambient dimension $n$, then for  any measurement matrix $A\in\mathds{R}^{d\times n}$,  any CS-recovery algorithm will result in $\ell_2$ reconstruction error of at least $1$. In other words, for any measurement matrix $A$, if $\hat{\xv}(\yv)$ denotes the reconstruction of $\xv\in\Qc$ from  measurements $\yv=A\xv$, then
\[
\inf_{\hat{\xv}: \mathds{R}^{d}\to \mathds{R}^n} \sup_{\xv} \|\hat{\xv}(\yv)- \xv \|_2 = 1.
\]
\end{lemma}
\begin{proof}
Consider measurement matrix $A\in\mathds{R}^{d\times n}$ with $d<n$ and some  reconstruction algorithm $\hat{\xv}: \mathds{R}^{d}\to \mathds{R}^n$.  Let ${\rm Ker}(A) \triangleq \{ \xv \ : \  A\xv =0\}$. Since $d<n$, ${\rm Ker}(A)- \{{\bf 0}\}\neq \emptyset$.  All signals in ${\rm Ker}(A) \cap \Bc_2^n(1)$ are mapped to the all-zero measurement vector, and hence the recovery algorithm maps all of them to some $\hat{\xv}\in\mathds{R}^n$. It is straightforward to confirm that
\[
\inf_{\hat{\xv}} \sup_{\xv \in {\rm Ker}(A) \cap \Bc_2^n(1)} \|\xv- \hat{\xv} \|_2 = 1.
\]
In fact the best reconstruction for $\xv \in {\rm Ker}(A) \cap \Bc_2^n(1)$ is $\hat{\xv} ({\bf 0})= {\bf 0}$, which leads to $\sup_{\xv \in {\rm Ker}(A) \cap \Bc_2^n(1)}\|\xv- \hat{\xv} \|_2 = 1$.
\end{proof}

\vspace{.2cm}

%

The answer to the second question is not as trivial and requires a more formal definition of ``success'' for the recovery algorithms. To address this question, we start with two formal definitions of applicability of CS to a compact set $\Qc$. We then derive a connection between these two notions and the rate-distortion performance of a code.
 \begin{definition}\label{def:unifappl} Compressed sensing is said to be {\em strongly applicable} to  compact set $\Qc \subset \mathds{R}^n$ with $d$ measurements, if, for any $\e>0$, there exists a $d \times n$ matrix $A_{\e}$ and a recovery algorithm $\Ac_{\e}$, 
 \[
 \Ac_{\e}: \mathds{R}^d\to \mathds{R}^n,
 \] 
 such that $\|\Ac_{\e}(A_{\e}\xv)-\xv\|_2\leq \e$, for all  $\xv\in\Qc$.  
\end{definition}

Another popular notion of applicability of CS to $\Qc$ is what we call {\em weak applicability} defined as follows.

\begin{definition} \label{def:indapplic}Compressed sensing is said to be {\em weakly applicable} to  compact set $\Qc \subset \mathds{R}^n$ with $d$ measurements, if, for any $\xv_o \in \Qc$ and and $\e>0$ there exists a $d \times n$ matrix $A$ such that for any $ \xv \in \Qc$ with $\|\xv_o-\xv\|_2>\e$,
\[
A \xv_o \neq A\xv.
\]
 \end{definition}
(Refer to \cite{TrGi07} and the reference therein for  some examples of weak CS-applicability.) Note the subtle difference between the two definitions. In Definition \ref{def:indapplic}, $A$ may depend on the the vector $\xv_o$. However, Definition \ref{def:unifappl} requires the existence of at least one $A$, that works on all $\xv_o \in \mathcal{Q}$. Therefore, as the name suggests strong CS-applicability is a stronger notion and hence, intuitively speaking, it requires more measurements. 

\begin{remark}
 The definition of weak CS-applicability might suggest that  it is impractical, as the measurement matrix can depend on $\xv_o$. However, as we will show later in the paper, for any $\xv_o \in \Qc$, random matrices satisfy the conditions required for weak CS-applicability with high probability.  
\end{remark}

\begin{remark}
CS-applicability is concerned with the recovery of $\xv_o$ from undersampled set of linear measurements $\yv_o = A \xv_o$. However, in many applications we require additional constraints on the system. Most notably, the system is usually required  to be robust to small variations on the measurements. In fact, since noise is an inevitable part of all measurement systems, in many systems $\yv_o = A \xv_o + \zv$, where $\zv$ denotes the measurement noise. Therefore, it is also important to ensure that the recovery algorithm is robust to the noise. We discuss this in more detail in Sections \ref{sec:measnoisei} and \ref{ssec:noise:ind}. 
\end{remark}

Our next step is to establish a connection between the rate-distortion performance of a code $\mathcal{C}_r$ on $\Qc$ and the number of measurements that makes CS-applicable to $\Qc$. The following definition plays a major role in this connection.

 \begin{definition}
Consider  compact set $\Qc \subset \mathds{R}^n$, and a family of fixed-rate compression codes, $\{(\Ec_r,\Dc_r): r>0\}$, with rate-distortion function $r(\d)$. Define the  $\a$-dimension   of  a family of codes  as\footnote{Note that distortion $\d$ here is defined in terms of the $\ell_2$-norm. In many papers, distortion is defined as the square of $\ell_2$-norm, a.k.a., square error. In those cases, $\alpha$-dimension shall be defined as $ \limsup_{\d \to 0} \frac{2r(\d)}{\log_2({1 \over \d})}$.}
 \begin{equation}
\alpha \triangleq \limsup_{\d \to 0} \frac{r(\d)}{\log_2({1 \over \d})}.
 \end{equation}
\end{definition}

In Section \ref{sec:related-work} we discuss the connection between $\a$-dimension and other well-known concepts in information theory and functional analysis such as Minkowski dimension and R{\'e}nyi entropy. 

\begin{example}\label{example:3}
Consider the family of compression algorithms presented in Example \ref{example:1} for $B_2^n(1)$. It is straightforward to confirm that the $\alpha$-dimension of this code is less than or equal to $n$. Furthermore, employing \eqref{eq:lowerrd} and lower bounds on $H_{\epsilon}$, one can show that the $\alpha$-dimension is also lower bounded with $n$ and hence is exactly equal to $n$.
\end{example}

\begin{example}\label{example:4}
Consider the family of compression algorithms presented in Example \ref{example:2} for $\Gamma_k^n \cap B_2^n(1)$. Similar to Example \ref{example:3} one can show that there exists a family of compression codes for $\Gamma_k^n \cap B_2^n(1)$ with   $\alpha$-dimension   equal to $k$.
\end{example}

It turns out that there is a close connection between the $\alpha$-dimension of a family of compression algorithms and CS-applicability.  Given $\a>0$, let $\Sc_{\a}^n$ denote the set of all  subsets of $\Bc_2^n(1)$ for which there exists a family of compression algorithms with $\a$-dimension upper-bounded by $\alpha$.  For each $\Qc \in \Sc_{\a}^n$,  define $d^w_{\min} (\Qc)$ ($d^s_{\min} (\Qc)$) as the minimum number of measurements required to make  CS  weakly (strongly) applicable to $\Qc$. The following theorems provide lower bounds on $d^w_{\min} (\Qc)$ and  $d^s_{\min} (\Qc)$. 
 
 \begin{proposition} \label{thm:lowerboundi}
 If CS is weakly applicable to any element of $\Sc_{\a}^n$ with $d$ measurements, then $d \geq \lfloor \alpha \rfloor$. In other words, 
 \[
 \sup_{\Qc \in \Sc_{\a}^n} d^w_{\min} (\Qc) \geq \lfloor \alpha \rfloor.
 \]
 \end{proposition}
 \begin{proof}
 Set $k \triangleq \lfloor \alpha \rfloor$ and define $\Theta_k$ as the  set of  vectors  in $\Bc_2^n(1)$, whose $n-k$ last coordinates are equal to zero, \ie
 \[
 \Theta_k \triangleq \{ \xv\in \Bc_2^n(1):  \ x_{k+1} = x_{k+2} = \ldots = x_{n}=0 \}.
 \]   
As shown in Example \ref{example:4}, $\Theta_k \in \Sc_{\a}^n$. Let $\xv_o \in \Theta_k$. Consider any measurement matrix $A \in \mathbb{R}^{d \times n}$ with $d < k$. Clearly, there are infinitely many other signals $ \xv \in \Theta_k$ that satisfy $A\xv = A\xv_o$. Hence, CS is not  weakly applicable to  $\Theta_k $ with $d < k$.  
 \end{proof}
\vspace{.1cm}
Our next theorem that rephrases Theorem 3 of \cite{DoEl03} derives a similar bound for the strong applicability. 
\vspace{.1cm}

 \begin{proposition} \label{thm:lowerboundu}
If CS is strongly applicable to any element of $\Sc_{\a}^n$ with $d$ measurements, then $d \geq \lfloor 2 \alpha \rfloor$. In other words, 
 \[
 \sup_{\Qc \in \Sc_{\a}^n} d^s_{\min} (\Qc) \geq \lfloor 2 \alpha \rfloor.
 \]
 \end{proposition}
 \begin{proof}
 Set $k \triangleq \lfloor \alpha \rfloor$ and define $\Delta^n_k = \Gamma^n_k \cap B_2(1)$, where  $\Gamma^n_k $ is the set of k-sparse vectors as defined in \eqref{eq:ksparse}.
As shown in Example \ref{example:4}, $\Delta_k^n \in \Sc_{\a}^n$.  Consider $d=\lfloor 2\a\rfloor -1$ and measurement matrix $A=[{\bf a}_1,\ldots,{\bf a}_n]\in\mathds{R}^{d\times n}$, with ${\bf a}_i$ denoting the $i^{\rm th}$ column of $A$.  We prove that for any recovery algorithm and any $\epsilon>0$, there exists a signal with reconstruction error greater than $\epsilon$. Since ${\rm rank}(A)\leq d$, corresponding to any $d+1$ columns of $A$, $i_1,\ldots,i_{d+1}\in\{1,\ldots,n\}$, there exists $(c_{i_1},\ldots,c_{i_{d+1}})\in\mathds{R}^{d+1}$ such that $\sum_{j=1}^{d+1} c_{i_j}{\bf a}_{i_j}=0$. Assume that $d+1$ is an even number. We construct two vectors  $\xv_1$ and $\xv_2$ both in $\mathds{R}^n$ such that $A\xv_1=A\xv_2$ and $\xv_1,\xv_2\in \Delta^n_k$. For $k\in\{i_1,\ldots,i_{(d+1)/2}\}$, let $x_{1,k}=c_k$ and set the rest of entries in  $\xv_1$ to zero. Similarly,  for $k\in\{i_{1+(d+1)/2},\ldots,i_{d+1}\}$, let $x_{2,k}=-c_k$ and set the rest of entries in  $\xv_2$ to zero. By our construction, $\xv_1$ and $\xv_2$, while having no intersection between their support sets,  are not distinguishable from their measurements. It is clear that for any $\beta>0$, $\beta \xv_1$ and $\beta \xv_2$ are not distinguishable from their measurements as well, \ie $A(\b \xv_1)=A(\b \xv_2)$. Let $\|\xv_1\|_2 \geq \| \xv_2\|_2$, and set $\beta$ such that $\| \beta \xv_1\|_2 =1$. Let $\xvh$ denote the reconstruction vector assigned to $A(\b \xv_1)$. By the triangle inequality,  $\|\beta \xv_1-\xvh\|_2+\|\beta \xv_1-\xvh\|_2\geq \|\beta \xv_1-\beta \xv_2\|_2$. On the other hand,  $\|\beta \xv_1 - \beta \xv_2 \|_2 = \| \beta \xv_1\|_2+ \| \beta \xv_2\|_2 \geq 1$. Therefore, either $\|\beta \xv_1-\xvh\|_2$ or $\|\beta \xv_2-\xvh\|_2$ is greater than $1/2$. For the case where $d+1$ is an odd number, the analysis is very similar.
 Hence, overall, we require $d\geq \lfloor 2\a\rfloor$ for strong CS-applicability. 
  \end{proof}

 Propositions \ref{thm:lowerboundi} and \ref{thm:lowerboundu}  provide lower bounds for the number of measurements that are required for CS-recovery method. However, it is not clear if these number of measurements are sufficient. In the rest of the paper, we show that considering random measurement matrices (nonadaptive measurements) and employing CSP, result in an accurate recovery algorithm with the number of measurements that are essentially the same as the ones proposed by Propositions \ref{thm:lowerboundi} and \ref{thm:lowerboundu}. 

\begin{remark}
As we will prove later, and as the lower bounds provided in  Propositions  \ref{thm:lowerboundi} and \ref{thm:lowerboundu} suggest,  the number of measurements that are required for strong CS-applicability is essentially two times that of weak CS-applicability. 
\end{remark}

Next section summarizes our results on the performance of CSP in recovering individual sequences. This section corresponds to weak CS-applicability. Section \ref{sec:unifCS} characterizes the performance of CSP for all vectors in $\mathcal{Q}$, which corresponds to strong CS-applicability concept we have introduced. We show that in each framework CSP is successful as long as the number of measurements is higher than the lower bounds  derived in Propositions \ref{thm:lowerboundi} and \ref{thm:lowerboundu}. 

\vspace{1cm}

\section{CSP recovery of individual sequences}\label{sec:contributions}

Consider the problem of recovering  signal $\xv_o\in\Qc \subset \mathds{R}^n$, from  $d <n $ linear measurements $\yv_o=A\xv_o+\zv$, where the entries of $A$ are i.i.d.~$\Nc(0,1)$, and $\zv\in\mathds{R}^d$ represents the measurement noise in the system. Furthermore, assume that there exists a family of  compression algorithms, $\{(\Ec_r,\Dc_r): r>0\}$, for the signals of $\Qc$ with rate-distortion  $r(\d)$. We employ the CSP algorithm described in \eqref{eq:CSP} to recover $\xv_o$ from $\yv_o$. The focus of this section is on the weak CS-applicability framework. Therefore, $\xv_o$ is considered to be  a fixed (but arbitrary) element of $\Qc$. 

\subsection{Noiseless measurements}

 Our first result is concerned with the performance of the CSP algorithm, when there is no noise in the system, \ie $\zv={\bf 0}$.

\begin{theorem}\label{thm:CSP}
Consider compression code $(\Ec,\Dc)$ for set $\Qc$ operating at rate $r$ and distortion $\d$. Let $A\in\mathds{R}^{d\times n}$, where $A_{i,j}$ are i.i.d.~$\Nc(0,1)$. For $\xv_o\in\Qc$, let $\xvh_o$ denote the reconstruction of $\xv_o$ from $\yv_o=A\xv_o$,  by the CSP algorithm employing code $(\Ec,\Dc)$. Then,
\[
\| \xvh_o- \xv_o\|_2 \leq \d\sqrt{1+\tau_1\over 1-\tau_2},
\]
with probability at least 
\[
1-2^{r}{\rm e} ^{\frac{d}{2}(\tau_2 + \log(1- \tau_2))}- {\rm e} ^{-\frac{d}{2}(\tau_1 - \log(1+ \tau_1))},
\]
where $\tau_1>0$ and $\tau_2\in(0,1)$ are arbitrary.\\
\end{theorem}
\noindent See Section \ref{proof:indcsapp} for the proof of this theorem.  

This theorem characterizes the trade-offs between the number of measurements, reconstruction error, and the probability of correct recovery. As is clear from the theorem, for a fixed $d$, as $\d \rightarrow 0$ the reconstruction error decreases to zero but the success probability of the algorithm decreases to zero as well. In fact, if we want the reconstruction error to converge to zero and at the same time the success probability remain close to $1$ (in high dimensional settings), the number of measurements should be larger than a certain number.  The next corollary characterize this number. \\

\begin{corollary}\label{corollary1}
Consider the setup of Theorem \ref{thm:CSP}, and let the number of measurements $d = \eta r/\log_2 (1/\ex \d)$, where $\eta >1$ is a parameter. Given $\e>0$, let $\d$ be such that
\[
{\eta \over \log{1\over  \ex \d}}<\e.
\]
Then, for $\xv_o\in\Qc$,
\[
\P(\|\hat{\xv}_o - \xv_o \|_2 \geq \theta {\d}^{1-(1+\e)/\eta}) \leq {\rm e}^{-0.8 d} + {\rm e}^{-0.3\e r},
\]
where $\theta=2\ex^{-(1+\e)/\eta}$.
\end{corollary}
\begin{proof}
In Theorem \ref{thm:CSP}, let $\tau_1=3$, $\tau_2=1-(\ex \d)^{2(1+\e)/\eta}$. For $\tau_1=3$, $0.5(\tau_1-\log(1+\tau_1))>0.8$. For $\tau_2=1-(\ex \d)^{2(1+\e)/\eta}$,
\begin{align}
&r\log 2+{d\over 2}(\tau_2+\log(1-\tau_2))\nonumber\\
&=r\log 2+{d\over 2}\Big(1-(\ex \d)^{2(1+\e)/\eta}+{2(1+\e) \over \eta}\log(\ex \d)\Big)\nonumber\\
& \leq r(\log 2 + \frac{\eta}{2 \log_2 (1/\ex \d)}- (1+ \epsilon) \log 2) \nonumber \\
&\overset{(a)}{\leq} r\log 2(1 +{\e \over 2}-(1+\e))\nonumber\\
&\leq -0.3 \e r,\label{eq:bound-p-details}
\end{align}
where $(a)$ is due to the fact that ${\eta \over \log_2{1\over  \ex \d}}= {\eta \log 2 \over \log{1\over  \ex \d}} <{\e \over 2}$.
Finally,
\begin{align}
\d\sqrt{1+\tau_1 \over 1-\tau_2} &=\d\sqrt{4\over (\ex \d)^{2(1+\e)/\eta}}\nonumber\\
&=\theta {\d}^{1-(1+\e)/\eta},
\end{align}
where $\theta=2\ex^{-(1+\e)/\eta}$.
\end{proof}

\vspace{.2cm}

\begin{remark}\label{remark:noiselessi1}
Let $\epsilon$ be a small positive number and set $d = \eta r/ \log_2 (1/\ex \d)$, with $\eta > 1+ \epsilon$. Then Corollary \ref{corollary1} states that as $\d \rightarrow 0$ the reconstruction error $\theta {\d}^{1-(1+\epsilon)/\eta} \rightarrow 0$, while the number of measurements converge to $\eta \alpha$, where $\alpha$ is the $\alpha$-dimension of the compression algorithms. In other words, as long as $d> \alpha$, CSP recovers $x_o$ accurately. Therefore, according to Definition \ref{def:indapplic}, CS is weakly applicable to $\Qc$ with $d$  measurements as long as $d > \alpha$. 
\end{remark}

\vspace{.1cm}

\begin{remark}
According to Proposition \ref{thm:lowerboundi}, any recovery algorithm based on the compression method requires at least $\lfloor \alpha \rfloor$ measurements for accurate recovery. Therefore, CSP  achieves the fundamental limit of signal recovery from compression algorithms.   
\end{remark}

\subsection{Noisy measurements}\label{sec:measnoisei}

So far we have considered the ideal setting, where there is no noise in the system. However, noise is an inevitable part of any sampling system and the robustness to measurement noise is a vital requirement for any recovery method. In this section we prove that CSP is robust to noise. Toward this goal, we consider two different types of measurement noise, stochastic and deterministic, and analyze the performance of CSP. Deterministic noise is considered as a good model for signal/measurement dependent noises such as quantization, while the stochastic noise models other noises such as amplifier noise in analog to digital converters. 

\subsubsection{Deterministic noise}\label{sec:detnoiseind}

 Consider the problem of recovering a vector $\xv_o$ from a noisy, undersampled set of linear measurements $\yv_o = A \xv_o + \zv$, where $A_{i,j}$ are i.i.d.~$\Nc(0,1)$ and $\zv$ denotes the noise. Let   $\|\zv \|_2 \leq \zeta$. Here, except for an upper bound on the $\ell_2$-norm, we do not make any other assumption on the noise. In particular the noise can be dependent on the measurement vector $\yv_o$.  Again we recover $\xv_o$ from $\yv_o$ by employing CSP algorithm described in \eqref{eq:CSP}. The following theorem provides a performance guarantee for the CSP algorithm: 

\begin{theorem}\label{thm:CSPdetnoiseind}
Consider compression code $(\Ec,\Dc)$  operating at rate $r$ and distortion $\d$ on set $\Qc$. For $\xv_o\in\Qc$, and $\yv_o=A\xv_o+\zv$ with   $\|\zv\|_2 \leq \zeta$,  let $\xvh_o$ denote the reconstruction of $\xv_o$ from $\yv_o$ offered by the CSP algorithm employing code  $(\Ec,\Dc)$. Then,
\[
\| \xvh_o- \xv_o\|_2 \leq \d\sqrt{1+\tau_1\over 1-\tau_2}+ \frac{2 \zeta}{\sqrt{(1- \tau_2)d}},
\]
with probability exceeding
\[
1-2^{r}{\rm e} ^{\frac{d}{2}(\tau_2 + \log(1- \tau_2))}- {\rm e} ^{-\frac{d}{2}(\tau_1 - \log(1+ \tau_1))},
\]
where $\tau_1>0$ and $\tau_2\in(0,1)$ are arbitrary.
\end{theorem}
\noindent See Section \ref{proof:indcsapp} for the proof. 

Again this theorem is quite general and depending on the number of measurements and $\zeta$, we can optimize the parameters to obtain the best bound. 
Note that since the size of the measurement noise is $\zeta$, and the system of linear measurements is underdetermined, we do not expect to recover $\xv_o$ with better accuracy than $\zeta$. Therefore, a proper choice for $\d$ is $\zeta$. The following corollary simplifies the statement of Theorem \ref{thm:CSPdetnoiseind} in this setting.  

\begin{corollary}\label{cor:detniseind}
Consider the setup of Theorem \ref{thm:CSPdetnoiseind} and let $d \geq \frac{\eta r }{\log_2 1/(\ex \d)}$, where $\eta> 1$. Assume that $\zeta=\d$ and 
\[
{\eta\over \log{1\over \ex \d}}<\e,
\]
for some $\e>0$. Then, for $\xv_o\in\Qc$,  
\[
\P(\| \xvh_o- \xv_o\|_2 \geq \theta' \zeta^{1-(1+\e)/\eta}) \leq {\rm e}^{-d/2} + {\rm e}^{- 0.3\e r},
\]
where $\theta'=2\ex^{-(1+\e)/\eta}+2/\sqrt{d}$.
\end{corollary}
\begin{proof}
If we choose $\tau_1=3$ and $\tau_2=1-(\ex \d)^{2(1+\e)/\eta}$ in Theorem \ref{thm:CSPdetnoiseind} the same analysis as the one presented in the proof of Corollary \ref{corollary1} yields the bound on the probability of error. Also, for this choice of variables, we have
\begin{align}
 &\d\sqrt{1+\tau_1\over 1-\tau_2}+ \frac{2 \zeta}{\sqrt{(1- \tau_2)d}}\nonumber\\
 &=\theta'\zeta^{1-(1+\e)/\eta},
\end{align}
where $\theta'=2\ex^{-(1+\e)/\eta}+2/\sqrt{d}$.
\end{proof}

\vspace{.2cm}

\begin{remark}
Consider the small noise regime, i.e., $\zeta \ll 1$. According to Corollary \ref{cor:detniseind}, as the number of measurement increases, or equivalently as $\eta$ increases, the reconstruction error decreases. However, it is always greater than or equal to $2 \zeta$ and gets closer to this bound as the number of measurement increases. The error $2 \zeta$ is not a fundamental bound on the signal recovery.  In fact, if we set $\tau_1$ to a small number greater than zero, the bound can be reduced to $\zeta$. 
\end{remark}

\begin{remark}
As $\eta$ decreases  the noise amplification increases. Furthermore, once $\eta$ reaches $1$ the reconstruction error is equal to $1$ and hence the scheme is not reliable any more.  In fact, according to corollary \ref{corollary1}, once we have more than $r(\d)/ \log_2 (\ex \d)$ measurements we can recover the signal accurately in the noiseless setting. However, if $\eta$ is close to $1$ the algorithm will not be robust to noise and adding a little bit of measurement noise leads to large reconstruction errors.  
\end{remark}

\vspace{.5cm}

\subsubsection{Stochastic noise}\label{sec:indstochnoise}
In the last section we studied the effect of deterministic noise. As is clear from our discussion, the assumptions on deterministic noise are  minimal. Therefore, in the analysis we should always consider the ``least favorable noise'' and derive the bounds for such pessimistic scenarios. For some types of  noise such as the quantization noise, the deterministic model seems to be a proper model. However, for many other noise sources stochastic model is a better match. Here, we consider the case where the linear measurements are corrupted by i.i.d.~noise, \ie  $\yv_o=A\xv_o+\zv$, where $z_i\sim\Nc(0,\s^2)$, $i=1,\ldots,d$, and analyze the performance of CSP under this model.

The first analysis we present here is based on  combining the results of Section \ref{sec:detnoiseind} and some probabilistic bounds on the $\ell_2$-norm of an i.i.d. Gaussian vector. It is straightforward to prove that $\P(\|\zv\|_2 \geq d\sigma (1+ \tau_3)) \leq {\rm e}^{- \frac{d}{2} (\tau_3- \log(1+ \tau_3))}$ (See section \ref{sec:proofback} for more information on this). Combining this result with Theorem \ref{thm:CSPdetnoiseind} immediately establishes the following theorem: 

\begin{theorem}\label{thm:csp_stocnoise1ind}
Consider the setup of Theorem \ref{thm:CSPdetnoiseind} with the only difference being  that noise $\zv$ is now drawn from $\Nc(0,\s^2I_d)$. Then,  for $\xv_o\in\Qc$,
\begin{align}
&\P \Big(\ \|\xv_o- \hat{\xv}_o\|_2 \geq {\d\sqrt{1+\tau_1}+2\s \sqrt{1+\tau_3}\over \sqrt{1-\tau_2}} \Big) \nonumber \\
& \leq  2^{r}{\rm e} ^{\frac{d}{2}(\tau_2 + \log(1- \tau_2))}+ {\rm e} ^{-\frac{d}{2}(\tau_1 - \log(1+ \tau_1))}\nonumber\\
&\;\;+ {\rm e}^{- \frac{d}{2} (\tau_3- \log(1+ \tau_3))},
\end{align}
where $\tau_2 \in (0,1)$, $\tau_1,\tau_3\geq 0$.
\end{theorem}

Setting  the free parameters in Theorem \ref{thm:csp_stocnoise1ind}  similar to Corollary \ref{cor:detniseind} yields  a similar conclusion. However, Theorem \ref{thm:csp_stocnoise1ind} has one counter-intuitive aspect: increasing the number of measurements does not reduce the reconstruction error. 
In other words, we expect (for fixed $n, \d$) the reconstruction error to reduce as the number of measurement increases. In fact, Theorem \ref{thm:CSPdetnoiseind} is proved for the ``least favorable'' noise, and it provides pessimistic bounds for the stochastic noise. Our next theorem settles this issue to some extent.  In the next theorem, we show that this is an artifact of the proof technique.


\begin{theorem}\label{thm:CSP-noisy}
Let $\xvh_o$ denote the solution of CSP to input $\yv_o=A\xv_o+\zv$, employing  code $(\Ec,\Dc)$ operating at rate $r$ and distortion $\d$.  Let $d={\eta r\over \log_2 {1\over {\rm e}\d}}$, where $\eta>1$,  and 
\begin{align}
{\eta \over \log {1\over \ex \d}} \leq \e'.\label{eq:bound-on-e-p}
\end{align}
for some  $\e'>0$. Let $\b=\sqrt{\log_2{1/\ex \d}}$. For any $\xv_o\in\Qc$, we have
\[
\| \xvh_o- \xv_o\|_2 \leq {1\over (\ex \d)^{(1+\e')/\eta}}({2\s \beta \over \sqrt{\eta}}+\sqrt{{4\s^2 \beta^2 \over \eta}+2{\d}^2+{4\s \d\over \sqrt{\eta}}}),
\]
with probability exceeding
\begin{align}
&1-2\ex^{-{0.15\eta r\over \log_2{1\over \ex \d}}}-\ex^{-{ r\over \log_2{1\over \ex \d}}} -\ex^{-0.3r}- {\rm e}^{- 0.3\e' r  }.
\end{align}

\end{theorem}
\noindent See Section \ref{proof:indcsapp} for the proof. \\

It is important to note the following two important aspects of this theorem: (i) As the number of measurements increases or equivalently as $\eta$ increases (for fixed $n, \d$) all the terms that are due to noise (the terms that have $\sigma$) decrease. Therefore, the reconstruction error decreases. (ii) This theorem is not sharp enough for small values of $\d$. In fact, as $\d \rightarrow 0$ the upper bound of the reconstruction error becomes weaker than that of Theorem \ref{thm:csp_stocnoise1ind}. We believe that the value of $\beta$ in the reconstruction error in Theorem \ref{thm:CSP-noisy} shall be proportional to the $\alpha$-dimension of the coder. This is left as an open question for future research.    

\section{CSP performance: strong CS-applicability}\label{sec:unifCS}
In the last section we explored the performance of CSP algorithm in recovering an individual sequence. The goal of this section is to extend our results to the strong-CS applicability problem. As before, we start with the noiseless setting and will then discuss the noisy measurements. 

\subsection{Noiseless measurements}
Consider  compact set $\Qc \subset \mathds{R}^n$ and  compression code $(\Ec,\Dc)$ for  $\Qc$ with rate $r$ and distortion $\d$.  Our goal is to recover $\xv_o$ from an underdetermined set of linear equations $\yv_o = A \xv_o$, where $A \in \mathds{R}^{d \times n}$ denotes a measurement matrix drawn as $ A_{ij} \overset{\rm i.i.d.}{\sim} \Nc(0,1)$. \\

\begin{theorem}\label{thm:finiteuniform}
 For  $\xv_o \in \mathcal{Q}$, let $\hat{\xv}_o$ denote the reconstruction of the CSP algorithm applied to  $\yv_o=A \xv_o$, when employing code $(\Ec,\Dc)$. We have
\begin{eqnarray*}
\lefteqn{\P\Big(\exists\; \xv_o \in \mathcal{Q} \ : \ \|\xv_o- \hat{\xv}_o\|_2 \geq \frac{2\d}{\sqrt{1-\tau}}(\sqrt{\frac{n}{d}}+ (1+t) ) \Big)} \nonumber \\
& \leq & {\rm e}^{-dt^2/2} + 2^{2r} {\rm e}^{\frac{d}{2}(\tau + \log(1- \tau))}, \hspace{3cm}
\end{eqnarray*}
where $\tau \in (0,1)$ and $t\geq 0$. 
\end{theorem}
\noindent See Section \ref{proof:unifcsapp} for the proof. \\

Compared to Theorem \ref{thm:CSP}, Theorem \ref{thm:finiteuniform} provides a stronger performance guarantee for CSP. It ensures that once a matrix is drawn, it will work on all $\xv_o \in \mathcal{Q}$. However, this strength has come at the price of larger reconstruction error and lower success probability. The following corollary presents a  more quantitive comparison between the two theorems. \\
 
\begin{corollary}\label{cor:metricdimension1}
Let $\d<1$ and $d = {2 \eta r \over \log_2 (1/\ex \d)}$, where $\eta > 1$, and
\[
{\eta\over \log{1\over \ex \d}} < \e.
\] Let $\theta\triangleq2(\sqrt{n/d}+ 2 )$. Then
\[
\P( \exists \xv_o \in \Qc : \|\hat{\xv}_o - \xv_o\|_2 \geq \theta {\d}^{1-(1+\e)/\eta}) \leq {\rm e}^{-d/2} + {\rm e}^{-0.6\e r}.
\]
\end{corollary}
\begin{proof}
Let $t=1$ and $\tau = 1-{\d}^{2(1+\e)/\eta}$ in Theorem \ref{thm:finiteuniform}. Then,
\begin{align*}
\frac{2\d}{\sqrt{1-\tau}}(\sqrt{\frac{n}{d}}+ (1+t) ) =2(\sqrt{\frac{n}{d}}+ 2 ){\d}^{1-(1+\e)/\eta}.
\end{align*}
Furthermore, for  $d = { 2\eta r \over \log_2(1/ \ex \d)}$, from \eqref{eq:bound-p-details}, it follows that
\begin{align*}
(2\log 2 )r+{d\over 2}(\tau + \log(1- \tau))\leq -0.6\e r. 
\end{align*}

 \end{proof}
 
 \vspace{.2cm}

\begin{remark}
Similar to remark \ref{remark:noiselessi1}, we  conclude that if the number of measurements $d$ is larger than $2 \alpha$, then CSP  accurately recovers signals in $\Qc$  from an undersampled set of $d$ linear measurements. Hence, CS is strongly applicable to $\Qc$ with $d$ measurements as long as $d > 2 \alpha$. 
\end{remark} 
\vspace{.1cm}

\begin{remark}
According to Proposition  \ref{thm:lowerboundu}, $2 \alpha$ is a lower bound on the number of measurements that are required for the exact recovery in the strong CS-applicability regime. Therefore, CSP achieves the fundamental limit of recovery from compression algorithms.    
\end{remark}

\subsection{Noisy measurements}\label{ssec:noise:ind}
In the last section, we considered the ideal setting where there is no measurement noise. The objective of this section is to present our results regarding the robustness of CSP to measurement noise. Note that here we are interested in strong CS-applicability framework. As before, we consider two different types of noise: (i) deterministic, and (ii) stochastic.  

\subsubsection{Deterministic noise}
Consider compact set $\Qc$ and a compression code $(\Ec,\Dc)$ on $\Qc$, operating  at rate $r$ and distortion $\d$. Let $A\in\mathds{R}^{d\times n}$, where $A_{i,j}$ are i.i.d.~$\Nc(0,1)$ and for $\xv_o\in\Qc$, $\yv_o=A\xv_o+\zv$, where the measurement noise satisfies  $\|\zv\|_2 \leq \zeta$. The following theorem characterizes the performance of the CSP algorithm: \\

\begin{theorem}\label{thm:csp_detnoise}
Let $\xvh_o$ denote the reconstruction of $\xv_o$ from $\yv$ by the CSP algorithm employing code $(\Ec,\Dc)$.  We have
\begin{eqnarray*}
\lefteqn{\P\Big(\exists\; \xv_o \in \mathcal{Q} \ : \ \|\xv_o- \hat{\xv}_o\|_2 \geq \frac{2c\d}{\sqrt{1-\tau}}  + \frac{2 \zeta}{\sqrt{d(1-\tau)} }\Big)} \nonumber \\
& \leq & {\rm e}^{-dt^2/2} + 2^{2r} {\rm e}^{\frac{d}{2}(\tau + \log(1- \tau))}, \hspace{3cm}
\end{eqnarray*}
where $\tau \in (0,1)$, $t\geq 0$, and $c = (\sqrt{\frac{n}{d}}+ (1+t) )$. 
\end{theorem}
See Section \ref{proof:unifcsapp} for the proof. \\

To obtain a better understanding of this theorem let $\d = \zeta$. The following corollary simplifies our main result in this setting. 

\begin{corollary}\label{cor:metricdimension1noisy}
Consider the setup of Theorem \ref{thm:csp_detnoise} and assume that $\d=\zeta$. Let $d = { 2\eta r \over \log_2 (1/\ex \d)}$, where $\eta > 1$ and ${\eta\over \log{1\over \ex \d}}<\e$, for some $\e>0$. Let  $\theta = 2(\sqrt{\frac{n}{d}}+2)+ 2/\sqrt{d}$. Then
\[
\P(\exists\; \xv_o \in \mathcal{Q}: \|\hat{\xv}_o - \xv_o\|_2 \geq \theta \zeta^{1 -(1+\e)/\eta} )\leq {\rm e}^{-d/2} + {\rm e}^{-0.6 \e r}.
\]
\end{corollary}

\begin{proof}

Let $t=1$, $\tau = 1-{\d}^{2(1+\e)/\eta}$  in Theorem \ref{thm:csp_detnoise}. 
Then,
\[
 \frac{2\d}{\sqrt{1-\tau}}(\sqrt{\frac{n}{d}}+ (1+t) ) + \frac{2\zeta}{\sqrt{(1- \tau)d}} 
= \theta {\d}^{1-(1+\e)/\eta}.
\]
Furthermore, for  $d = { \eta r \over \log_2(1/ \ex {\d})}$, as shown in  \eqref{eq:bound-p-details}, 
\[
(2\log 2 )r+{d\over 2}(\tau + \log(1- \tau))\leq -0.6\e r.
\]
 \end{proof}
\vspace{.2cm}

\begin{remark}
Let $\zeta \ll 1$. Then the reconstruction error has the same order, i.e., $\zeta^{1-1/\eta}$  as Corollary \ref{cor:detniseind}. However, note that the number of measurements that are required in Corollary \ref{cor:metricdimension1noisy} is two times that of Corollary \ref{cor:detniseind}. \\
\end{remark}

\subsubsection{Stochastic noise}

Consider the problem of recovering a signal from an undersampled set of linear measurements in the presence of stochastic noise, \ie  $\yv_o=A\xv_o+\zv$, where $z_i\sim\Nc(0,\s^2)$, $i=1,\ldots,d$. To recover signal $\xv_o$ from $\yv_o$, again we employ the CSP algorithm described in \eqref{eq:CSP}. Parallel to our discussion in Section \ref{sec:indstochnoise}, we can apply Theorem \ref{thm:csp_detnoise} and obtain the following result: \\

\begin{theorem}\label{thm:csp_stocnoise1}
Consider the setup of Theorem \ref{thm:csp_detnoise} with the only difference being  that noise $\zv$ is now drawn from $\Nc(0,\s^2I)$. Then
\begin{align}
&\P \Big( \exists \;\xv_o \in \Qc: \ \|\xv_o- \hat{\xv}_o\|_2 \geq {2c{\d}+2\s \sqrt{1+\tau'}\over \sqrt{1-\tau}} \Big) \nonumber \\
& \leq\;  2^{r}{\rm e} ^{\frac{d}{2}(\tau + \log(1- \tau))}+ {\rm e} ^{-\frac{dt^2}{2}}+ {\rm e}^{- \frac{d}{2} (\tau'- \log(1+ \tau'))},
\end{align}
where $\tau \in (0,1)$, $\tau' \geq 0$ and $c = \sqrt{\frac{n}{d}}+ t+1 $.
\end{theorem}

However, similar to Theorem \ref{thm:csp_stocnoise1ind} this theorem suffers from an issue and that is the independence of  the reconstruction error from the number of measurements (for fixed $n, {\d}$). To resolve this issue we provide a finer analysis that captures the randomness of the noise more efficiently. \\

\begin{theorem}\label{thm:stocnoiseunif}
Consider compact set $\Qc\subset \mathds{R}^n$ and a compression code $(\Ec,\Dc)$ for set $\Qc$ at rate $r$ and distortion ${\d}$. Let $\xvh_o$ denote the reconstruction of CSP from noisy measurements  $\yv= A\xv_o + \zv$, where $\zv \sim \Nc(0,\s^2I)$. Then, the probability that there exists $\xv_o\in\Qc$ such that
\[
\| \xv_o- \hat{\xv}_o\|_2 \geq \frac{2(\sqrt{n} + (t+1) \sqrt{d}) {\d} + 2 \gamma \sigma }{ \sqrt{(1- \tau)d}}+{\d},
\]
is smaller that
\[
\ex^{-{dt^2\over 2}}+2^{2r}(\ex^{-\g^2/2}+\ex^{{d\over 2}(\tau+\log(1-\tau))}),
\]
where $\tau \in (0,1)$, and $t, \gamma >0$. 
\end{theorem}
See Section \ref{proof:unifcsapp} for the proof.

Our discussions after Theorem \ref{thm:csp_stocnoise1ind} also hold for this theorem. In fact, we believe that it is still possible to obtain sharper bounds on the reconstruction error of $\|\xv_o - \hat{\xv}_o\|_2$. This remains open for future research. 
 

\section{Extension to analog signals} \label{sec:extension}

\subsection{Analog CS}
So far, we have considered the problem of recovering finite-dimensional signals from their undersampled set of linear measurements. However, the framework we have developed for weak CS-applicability can be extended to infinite-dimensional spaces as well.\footnote{Our results on strong CS-applicability can not be extended to analog framework. Characterizing the reconstruction error of CSP in this case seems to be challenging and is left for the future research.} In this section, we extend our results to recovering continuous-time function $f: [0,1]\rightarrow \mathds{R}$ from a finite number of random linear measurements. The importance of this topic for many application areas, e.g. spectrum sensing, has made it the scope of extensive research. (For more information see \cite{TrLaDuRoBa10, MiElDoSh11, MiEl09} and the references therein.) Since our measurement and reconstruction techniques are different from the other work in the literature, we first review the related basic  concepts  required for analyzing   continuous-time functions and then extend CSP for recovering such signals. 

\subsection{Ito's integral}
For continuous-time signals, we consider a measurement system that is based on the Wiener process. Wiener process $W(t)$, a.k.a. Brownian motion, is a continuous time process that satisfies the following four properties:
\begin{enumerate}
\item $\P(W(0) = 0)=1$.
\item The probability that a randomly generated path to be continuous is equal to 1. 
\item $W(t)-W(s) = \Nc(0,t-s)$, for $0\leq s < t \leq 1$.
\item For $0 \leq s_1 < t_1 \leq s_2 < t_2\leq 1$, $W(t_1)-W(s_1)$ is independent of $W(t_2)-W(s_2)$.
\end{enumerate} 
This process is a key component of stochastic calculus and stochastic differential equations. In particular, Ito's integral, which plays a  central role in stochastic differential equations, is defined based on the Wiener process. To keep our discussions simple, we  introduce a specific form of the Ito's integral that is used in this paper. 

For  function $f: [0,1] \to \mathds{R}$, define its $p$-norm as
\[
\|f\|_p \triangleq \left({\int_{0}^{1} |f(t)|^p dt}\right)^{1/p}.
\]
Furthermore, define $L_p([0,1])$ as the set of functions from $[0,1]$ to $\mathds{R}$ with finite $p$-norm, \ie
\[
L_p([0,1]) \triangleq \{ f:[0,1] \rightarrow \mathds{R} \ | \ \|f\|_p < \infty\}. 
\]
In this paper we are mainly interested in $L_2([0,1])$, which is defined as  the set of  functions with finite second moment. Suppose that $f_s\in L_2([0,1])$ is a simple function, \ie $f_s$ can be represented as 
\[
f_s(t) = \sum_{k=1}^N c_k \ind_{t\in (t_k, t_{k+1}]},
 \]
 where $0= t_1<t_2< \ldots < t_N=1$, and $(c_1,\ldots,c_N)\in\mathds{R}^N$.
 For such functions, Ito's stochastic integral is defined as
\[
\int_{0}^1 f_s(t) dW(t) \triangleq \sum_{k=1}^N c_k(W(t_{k+1})- W(t_k)).
\]
Note that since $(W(t_{i+1})- W(t_i): i=0,1,\ldots,N-1)$ are independent Gaussian random variables, the result of this integral is a Gaussian random variable with mean zero and variance $\sum_{k=1}^N c_k^2 (t_{k+1}-t_k)$. For $f \in L_2([0,1])$, let $(f_1, f_2, \ldots)$ be a sequence of simple functions such that
\[
\lim_{n \rightarrow \infty} \int_{0}^1 (f(t)-f_n(t))^2 dt = 0.
\]
Then the Ito's integral of $f$ is defined as
\[
\int_{0}^1 f(t)dW(t) \triangleq \lim_{n \rightarrow \infty} \int_{0}^1 f_n(t) dW(t), 
\]
where the convergence is in the mean square sense.  The following theorem establishes the existence of the Ito integral and its final distribution in our setting: \\

\begin{theorem}\cite{shreve2004stochastic}\label{thm:gaussian_wiener}
Let $f \in L_2([0,1])$. Then $\int_{0}^1 f(t)dW(t)$ is normally distributed with mean zero and variance $\|f\|_2^2.$ 
\end{theorem}
The proof can be found on page 149 of \cite{shreve2004stochastic}.

\subsection{Distortion-rate function}

Consider a class of functions $\Fc \subset L_2([0,1])$, and a family of compression algorithms $\{(\Ec_r,\Dc_r): r>0\}$, indexed by rate $r$. For each code in this family, the encoder and decoder mappings, $(\Ec_r,\Dc_r)$, are defined as
\[
\Ec_r: \Fc\to\{1,2,\ldots,2^r\},
\]
and
\[
\Dc_r: \{1,2,\ldots,2^r\}\to \hat{\Fc},
\]
respectively, where $\hat{\Fc}\subset L_2([0,1])$ denotes the class of reconstruction functions.
For a function $f \in \Fc$, $\Dc_r(\Ec_r(f))$ denotes the reconstruction of function $f$ by the code $(\Ec_r,\Dc_r)$. Given  compression algorithm $(\Ec_r,\Dc_r)$, let $\Cc_r$ denote its {\em codebook} defined as  
\[
\Cc_r \triangleq \{ \Dc_r(\Ec_r(f)): f \in \mathcal{F} \}.
\]
 The  distortion-rate function of this family of codes is defined as
\[
\d(r) \triangleq \sup_{f \in \mathcal{F}} \|f- \Dc_r(\Ec_r(f))\|_2. 
\]

\subsection{Compressed sensing of analog signals}

\vspace{.1cm}

\subsubsection{Measurement process}
Unlike the classical compressed sensing setup, where the measurement process is assumed to be in the discrete time domain, here we consider analog domain measurements. In particular, for  function $f_o \in L_2([0,1])$, we consider $d$ linear measurements of the form 
\begin{equation}\label{eq:ito}
y_{o,i} = \int_0^1 f_o(t) dW_i(t),   \ \ \  {\rm for} \ \  i=1,2, \ldots, d,
\end{equation}
where $W_i$,  $ i=1,2, \ldots, d$, are independent Wiener processes. Similar to the discrete time settings, each measurement is a random linear combination of the signal at different times. As we will show in this section, this type of measurement process ensures that with ``sufficient'' number of measurement, the ``critical information'' about the signal is acquired by the measurements, and therefore, we can recover $f_o$ from the measurement vector $\y_o\in\mathds{R}^d$.

\vspace{.1cm}

\subsubsection{CSP algorithm}

Consider    a family of compression algorithms $\{(\Ec_r, \Dc_r): r>0\}$ for class of functions $\Fc\subset L_2([0,1])$ with  rate-distortion function $\d(r)$. 
We are interested in recovering a function $f_o \in \Fc \subset L_2([0,1])$ from  $d$ linear measurements $\yv_o\in\mathds{R}^d$, where for $i=1,\ldots,d$,
\begin{equation}
y_{o,i} \triangleq   \int_0^1 f_o(t) dW_i(t).\label{eq:calAdef}
\end{equation}
Let $\Ac: L_2([0,1])\to\mathds{R}^d$ denote the just-defined  linear measurement process, \ie $\yv_o=\Ac(f_o)$.
To recover the function $f_o$ from $\yv_o$, we employ the CSP algorithm defined as
\begin{align}
\hat{f}_o = \argmin_{\fh \in \mathcal{C}_r} \| \yv_o - \mathcal{A}(\hat{f})\|_2^2.\label{eq:CSP-cont}
\end{align}

The intuition for the CSP algorithm is the same as what we proposed before; among all the low-complexity signals (defined according to the compression algorithm) look for the one that matches the measurements the best. The parameter $r$ can be considered as a free parameter in the algorithm, whose role will be clear in the next section. 

\vspace{.1cm}

\subsubsection{Performance guarantees for CSP}

Consider the problem of recovering  function $f_o \in \mathcal{F} \subset L_2([0,1])$ from its undersampled  set of $d$ random linear measurements, $\y_o = \mathcal{A}(f_o)$, as defined in \eqref{eq:calAdef}. Assume that there exists a family of compression algorithms $(\Ec_r,\Dc_r)$  for $\mathcal{F}$ indexed by $r$,  that achieves the rate distortion function $r(\d)$. We employ the CSP algorithm to recover $f_o$. The following theorem characterizes the performance of the CSP algorithm.

\begin{theorem}\label{thm:CSPInf}
 For $f_o \in \Fc$, let $\hat{f}_o$ denote the reconstruction of $f_o$ from $\yv_o=\Ac(f_o)$, by the CSP algorithm employing rate-$r$ compression algorithm $(\Ec,\Dc)$. Then,
\[
\| \hat{f}_o- f_o\|_2 \leq {\d}\sqrt{1+\tau_1\over 1-\tau_2},
\]
with probability at least 
\[
1-2^{r}{\rm e} ^{\frac{d}{2}(\tau_2 + \log(1- \tau_2))}- {\rm e} ^{-\frac{d}{2}(\tau_1 - \log(1+ \tau_1))}.
\]
\end{theorem}
See Section \ref{proof:analogcs} for the proof.

Theorem \ref{thm:CSPInf} considers noiseless measurements.  In the case of noisy measurements,   assume that
\[
y_{o,i} = \int_0^1 f_o(t)dW_i(t) + z_i,
\]
where $z_i$ represents the measurement noise. The results we had in Section \ref{ssec:noise:ind} can be extended to the infinite-dimensional setting. Hence, CSP is robust to noise. For the sake of brevity, we only mention the result for the deterministic noise here.

\begin{theorem}\label{thm:CSPinf_noisy}
Consider compression code $(\Ec,\Dc)$ for set $\Fc$ operating at rate $r$ and distortion $\d$. For $f_o\in\Fc$, let $A\in\mathds{R}^{d\times n}$, where $A_{i,j}$ are i.i.d.~$\Nc(0,1)$, and $\yv_o=\mathcal{A}(f_o)+\zv$, where the measurement noise satisfies  $\|\zv\|_2 \leq \zeta$.  Let $\hat{f}_o$ denote the reconstruction of $f_o$ from $\yv_o$, by the CSP algorithm employing code  $(\Ec,\Dc)$. Then,
\[
\| \hat{f}_o- f_o\|_2 \leq {\d}\sqrt{1+\tau_1\over 1-\tau_2}+ \frac{2 \zeta}{\sqrt{(1- \tau_2)d}},
\]
with probability exceeding
\[
1-2^{r}{\rm e} ^{\frac{d}{2}(\tau_2 + \log(1- \tau_2))}- {\rm e} ^{-\frac{d}{2}(\tau_1 - \log(1+ \tau_1))},
\]
where $\tau_1>0$ and $\tau_2\in(0,1)$ are arbitrary.\\
\end{theorem}

After employing Theorem \ref{thm:gaussian_wiener}, the proof of this result is similar to the proof of Theorem \ref{thm:CSPdetnoiseind}. Hence we skip the proof. 

\begin{remark}
Comparing Theorems \ref{thm:CSP} and \ref{thm:CSP-noisy} with Theorem \ref{thm:CSPInf} and \ref{thm:CSPinf_noisy} may lead us to a conclusion that the ambient dimension of the signal is not important in the performance of CSP algorithm. This conclusion is true at the level of weak CS-applicability. However, it does not hold for strong CS-applicability. 
\end{remark}

\subsection{Applications}

In this section, we investigate the implications of Theorem \ref{thm:CSPInf} for three different classes of continuous time signals. As we will see in these examples, in case of infinite dimensional signals, the  rate-distortion performance shows more diverse types of behavior. Different rate distortion behaviors of these classes clarify the opportunities and limitations of analog CS. While our main focus in this section is on the noiseless signal recovery, one may employ Theorem \ref{thm:CSPinf_noisy} for each example and derive bounds for the reconstruction error in the presence of the noise.  

Let $\mathcal{P}_N^Q(\Delta)$ denote the class of piecewise polynomial functions with $N$, $Q$, $\Delta$ representing the maximum degree of the polynomials, number of singularity points\footnote{Singularity point is a point at which the signal is not infinitely differentiable}, and maximum value of the function, respectively.
\begin{example}\label{ex:piecepoly}
 There exists a family of compression algorithms for $\mathcal{P}_N^Q(\Delta)$  that achieves 
\[
r(\d)= (N+3)(Q+1) \log \left(\frac{1}{{\d}}\right)+c,
\]
where $c$ is a constant that depends on  $N, Q$ and $\Delta$, but is independent of ${\d}$ (Theorem 2.2 in \cite{MaCa08}).  
\end{example}
The  rate-distortion behavior described in Example \ref{ex:piecepoly} for $\mathcal{P}_N^Q(A)$ is reminiscent of the rate-distortion function of subsets of finite-dimensional spaces. However, since the locations of singularities are not fixed, $\mathcal{P}_N^Q(\Delta)$ is not a subset of any finite dimensional subspace of $L_2([0,1])$. Nevertheless, we would expect to recover the signals of this class with finite number of measurements with small error. In fact since the $\alpha$-dimension is equal to $(N+3) (Q+1)$, the signals of this class can be essentially recovered from $(N+3)(Q+1)$ measurements. Note that in our discussion we consider the noiseless setting. Otherwise, we may need more measurements to ensure robustness to noise.  \\

 For finite-dimensional spaces it is straightforward to show that for any compact subset of $\mathds{R}^n$ there exists a compression algorithm whose rate distortion function satisfies $r(\d) = O\left(n \log\left( \frac{1}{\d}\right) \right)$ (Any compact set can be covered by a ball of certain radius. Combining this fact with Example \ref{example:1} establishes the result).  However, in infinite dimensional spaces this is not the case any more. The next example illustrates a class with a slightly different rate-distortion behavior.

Let $\mathcal{H}_h(C)$ be a class of functions $f: \mathds{C} \rightarrow \mathds{C}$ satisfying the following properties:
\begin{itemize}
\item[a)] $f$ is analytic on a strip of size $h$, i.e., $f(z)$ is analytic on $\{z= x+iy \ | \ |y|\leq h \}$. 
\item[b)] $|f(z)|$ is bounded by  $C$. 
\end{itemize}
Define $\mathcal{G}_h(C)$ as
\[
\mathcal{G}_h(C) \triangleq \{ g: [0,1] \rightarrow \mathds{R}  : \exists\; f \in \mathcal{H}_h(C),  g(x) = |f(x+i0)|\}. 
\]   

\begin{example}\label{ex:analytic}
There exists a family of compression algorithms for $\mathcal{G}_h(C)$ that achieves
\[
r(\d) \leq c \left( \log \frac{1}{\d} \right)^2,
\]
where $c$ is a constant that does not depend on ${\d}$ \cite{Koleps59}.
\end{example}
Clearly for this class of functions the $\a$-dimension is infinite, therefore we do not expect to recover the signals accurately from finite number of measurements. However, CSP algorithm is still useful for this class as is described in the next corollary.

 \begin{corollary}
 Let ${\d}\ll1$ and $d = 4c \log_2(1/{\d})$. Then
\[
\P(\|\hat{\xv}_o - \xv_o\|_2 \geq \sqrt{2{\d}}) \leq {\rm e}^{-c\log(2) \log^2(1/{\d})}+ {\rm e}^{-{c \log_2(1/{\d}) \over 2}}.
\]
\end{corollary}
To prove this result, set $\tau_1 =1$ and $\tau_2 = 1-{\d}$ in Theorem \ref{thm:CSPInf}. While the number of required measurements  tends to infinity, as the distortion goes to zero,  it  only grows  logarithmically with the distortion. Therefore, intuitively speaking, accurate estimates are still obtained from few measurements. 

If the class of functions is too rich (less structured), then the growth rate of  rate-distortion will be faster and therefore to obtain reasonably accurate reconstruction we may require many observations. Here, we present one such example.

 Let $H^{\beta}(C)$ be the class of real functions $f:[0,1] \rightarrow \mathds{R}$, having derivative of order $\beta$ in $L^2$ (in the sense of Riemann-Liouville) strongly bounded by some constant $C$.
\begin{example} \label{ex:smooth}
 There exists a family of compression algorithms for $H^{\beta}(C)$  that achieves 
\[
r(\d)= c \left(\frac{1}{\d}\right)^{\frac{1}{\beta}},
\]
where $c$ is a constant independent of ${\d}$ \cite{Koleps59}.
\end{example}
According to Theorem \ref{thm:CSPInf},  having $O( \frac{1}{{\d}^{1/\beta} \log(1/{\d})})$ measurements,  the reconstruction error is bounded  by $\sqrt{2{\d}}$. It is clear that as the class of signals becomes richer CSP requires more measurements to achieve the same accuracy.


\section{Related work}\label{sec:related-work}

\subsection{Connection of compression and compressed sensing}
In this paper we consider the problem of using a family of compression algorithms for compressed sensing. The other direction, i.e., using CS for compression have also been extensively studied in the literature \cite{SaBaBa06, GoFlRa08, BoBa07, candes2006encoding, dai2009distortion, sun2009optimal, deng2010robust, schulz2009empirical, boufounos2012universal}. In this line of work the rate-distortion that is achieved by scaler (or in a few cases adaptive) quantization of random linear measurements has been derived. However, such results are different from our work since they only consider either sparse or approximately sparse signals. Furthermore, we consider a different direction, that is, the direction of deriving CS recovery algorithms based on compression schemes.

\subsection{Kolmogorov's $\epsilon$-entropy and embedology} \label{sec:relworkepsilon}

It is clear that our results can be stated in terms of Kolmogorov's $\epsilon$-entropy by considering it as the  optimal compression scheme from the perspective of  rate-distortion tradeoff.  For $\e>0$,  let $\Cc_{\e}$ denote an $\e$-covering of $\Qc\subset\mathds{R}^n$ such that $\log |\Cc_{\e}|=H_{\epsilon}(\Qc)$, and assume that  
 \[
 \limsup_{\epsilon \rightarrow 0}\frac{H_{\epsilon}(\mathcal{Q})}{\log(1/\epsilon)} \leq \alpha.
 \]
This quantity is called upper metric dimension \cite{Koleps59}, Minkowski dimension \cite{falconer2003fractal}, or even box-counting dimension \cite{BeEdFoNi93}.  Metric dimension is a measure of the massiveness of compact sets in finite dimensional spaces \cite{Koleps59}. Furthermore, this quantity is proved to be useful in the analysis of dynamical systems. 
In particular, started with the seminal work of Ma{\~n}{\'e} \cite{Mane81} many authors have explored the connection between the metric dimension of a set and the invertibility of its linear projections.\footnote{The reason this problem has been explored in the field of dynamical systems is that, when an attractor (attractor is a set towards which a variable in a dynamical system moves over time) is measured experimentally, what is actually observed is a projection or embedding of the attractor into a Euclidean space. Therefore, the major question is how accurately we can explore the properties of the attractor from its image.}  To have a detailed comparison between our work and embedology, we review the following theorem from \cite{BeEdFoNi93}:

\begin{theorem} \cite{BeEdFoNi93} \label{thm:holdercontinuous}
Let $\Qc \subset \mathbb{R}^n$ be a compact set with box dimension $\alpha$, and $P_0$ an orthogonal projection of rank $d > 2 \alpha$. Then for every $\e$ and $\theta \in (0, 1- 2\alpha/d)$, there exists an orthogonal projection $P$  of rank $d$, and a constant $C$ such that $\| (P-P_0) \xv \|_2 \leq \e \|x\|_2$ for all $\xv \in \mathbb{R}^n$ and,   $\|\xv_1 -\xv_2\|_2 \leq C \|P(\xv_1- \xv_2) \|_2 ^\theta$, $\forall \xv_1,\xv_2\in\Qc$. 
\end{theorem}  

Note that $\|\xv_1 -\xv_2\|_2 \leq C \|P(\xv_1- \xv_2) \|_2 ^\theta$ immediately implies that the inverse of $P$ exists if we restrict ourselves to the vectors in 
\[
P(\mathcal{Q}) \triangleq \{P(\xv) \ : \  \xv \in \Qc  \}.
\]
If we call this inverse mapping $P^{-1}$, then it is straightforward to obtain $ \|P^{-1} (\mathbf{u}) - P^{-1}(\mathbf{v}) \|_2 \leq C \| \mathbf{u} - \mathbf{v}\|_2^{\theta}$. In other words, Theorem \ref{thm:holdercontinuous} implies that not only we can recover the signal from its low-rank projection, but also the inverse mapping is H\"{o}lder-continuous. These results can be compared with Corollaries \ref{cor:metricdimension1} and \ref{cor:metricdimension1noisy} in our paper if we consider the $\epsilon$-entropy instead of the rate-distortion function. These two corollaries also show that if the number of measurements $d> 2 \alpha$, then CSP can accurately recover the data from $d$-measurements and that the recovery is stable. However, there are several major differences:
\begin{enumerate}
\item We consider a broad class of compression algorithms with generic rate-distortion performance. 
\item Our approach is constructive; we propose CSP that can recover the data even if we have access to only one compression algorithm with certain rate-distortion performance. The ``embedology'' literature is not constructive and only proves the existence of the inverse map.
\item Our approach has enabled us to explore the trade-off between the probability of correct recovery and the reconstruction error.
\item Our stability analysis is more elaborate and more general than the H\"{o}lder-continuity of the inverse mapping $P^{-1}$ for the following reasons: First, the inverse mapping is only defined on the set $P(\Qc)$. Hence if the noise moves  the measurements out of $P(\Qc)$, then the inverse mapping is not even defined. While the CSP algorithm is robust to deterministic noise, we can also provide accurate analysis of its performance in the presence of stochastic noise, as discussed in Theorem \ref{thm:CSP-noisy}. 
\end{enumerate}

Several extensions of Theorem \ref{thm:holdercontinuous} have been  explored in the literature. For instance \cite{HuKa99} has extended this result to infinite-dimensional Hilbert spaces and proved that if $d \geq 2\alpha$ almost every bounded linear transformation of $\Qc$ is invertible. They have also shown that the inverse mapping is H\"{o}lder, however, their H\"{o}lder exponent is much lower than Theorem  \ref{thm:holdercontinuous}; $\frac{d- 2 \alpha}{d(1+ \alpha)}$ versus $1- 2\alpha/d$ in Theorem \ref{thm:holdercontinuous}. For more information on this line of research refer to \cite{SaYoCa91, Mane81, HuKa99, BeEdFoNi93} and references therein. 

  The connection between Minkowski dimension and CS has also been explored in the stochastic settings, and we will review this connection in the next section.

\subsection{Stochastic settings} \label{sec:stochastic}
This paper considers a deterministic signal model. However, stochastic settings have also been considered  in CS \cite{HeCa09, Schniter2010, BaGuSh09, RaFlGo10, DoMaMoNSPT, DoMaMo09, DoTa09, WuVe10, MaDo09sp, WuVeOPtimal12, donoho2012information}. In such models the data is assumed to follow a certain distribution (often i.i.d.) and the probability of correct recovery is measured as the ambient dimension tends to infinity. In many cases the algorithms exhibit certain phase transitions in the probability of correct recovery. Such phase transitions have been characterized in certain cases either theoretically or empirically \cite{MaDo09sp, DoMaMoNSPT, DoMaMo09, WuVe10, WuVeOPtimal12}. 

The most relevant to our work are \cite{WuVeOPtimal12, WuVe10}. These two papers characterize the performance of ``information-theoretically'' optimal algorithms in the asymptotic setting. For instance they prove that the number of measurements that are required for ``exact'' recovery is the same as the R{\'e}nyi information dimension. Even though there is an interesting connection between R{\' e}nyi information dimension and metric dimension \cite{kawabata1994rate},  there are several major differences between our work and the work of \cite{WuVeOPtimal12, WuVe10}.
First our framework is concerned with the deterministic signal models. Second, our results are for finite-dimensional signals, and are non-asymptotic. Third, we consider arbitrary family of compression algorithms and characterize when such schemes can be used for signal recovery from random linear measurements.

\subsection{Kolmogorov complexity}
Our work is mainly inspired by series of work on the connection between Kolmogorov complexity of sequences and CS \cite{jalali2012minimum,kolmogrov_sampler, DoKaMe06, BaDu12, BaDu11, jalali2011minimum, JaMaBa12}. In particular, \cite{jalali2012minimum} defines the {\em Kolmogorov information dimension} of $\xv = (x_1, x_2, \ldots, x_n)\in[0,1]^n$ at resolution $m$ as
\[
\kappa_{m,n}(\xv) \triangleq \frac{K^{[\cdot]_{m}}(x_1,x_2, \ldots, x_n)}{m},
\]
where intuitively speaking, $K^{[\cdot]_{m}}(x_1,x_2, \ldots, x_n)$ denotes the Kolmogorov complexity of vector $\xv$, when each of its components is quantized by $m$ bits, 
and proves that if the Kolmogorov information dimension of a sequence is small compared to its ambient dimension one can recover it from an undersampled set of linear measurements. Our results have several connections with \cite{jalali2012minimum}. The proof techniques we use here have similarities to the proof techniques used in \cite{jalali2012minimum}. However, the problems are different. We believe our results in this paper present the first step in a new direction toward practical implementation of \cite{jalali2012minimum}. While the CSP algorithm is based on exhaustive search among the codewords at this point (since it is based on exhaustive search),  it provides an approach to designing  sub-optimal algorithms such as greedy methods. Furthermore, CSP algorithm may enable us to employ universal compression algorithms \cite{jalali2008rate, jalali2012lossy} and develop universal compressed sensing methods. This has been the main goal of \cite{jalali2012minimum, DoKaMe06, BaDu12, BaDu11, jalali2011minimum, JaMaBa12}.


\section{Proofs} \label{sec:proofs}

\subsection{Background}\label{sec:proofback}
We use the following two lemmas from \cite{JaMaBa12} throughout our proofs.

\begin{lemma}[$\chi^2$-concentration]\label{lemma:chi}
Fix $\tau>0$, and let $Z_i\sim\Nc(0,1)$, $i=1,2,\ldots,d$. Then,
\begin{align*}
\P\Big( \sum_{i=1}^d  Z_i^2 <d(1- \tau) \Big)  \leq {\rm e} ^{\frac{d}{2}(\tau + \log(1- \tau))}
\end{align*}
and
\begin{align}\label{eq:chisq}
\P\big( \sum_{i=1}^d  Z_i^2 > d(1+\tau) \Big)  \leq {\rm e} ^{-\frac{d}{2}(\tau - \log(1+ \tau))}.
\end{align}
\end{lemma}

\begin{lemma}\label{lemma:gaussian-vectors}
Let $\Xv$ and $\Yv$ denote two independent Gaussian vectors of length $n$ with i.i.d.\ elements. Further, assume that for $i=1,\ldots,n$,  $X_i\sim\Nc(0,1)$ and  $Y_i\sim\Nc(0,1)$. Then the distribution of  $\Xv^T\Yv=\sum_{i=1}^nX_iY_i$ is the same as the distribution of $\|\Xv\|_2G$, where $G\sim\Nc(0,1)$ is independent of $\|\Xv\|_2$.
\end{lemma}

\subsection{Calculation of  rate-distortion function}\label{proof:examples}
In this section we briefly summarize the proof of Example 1 and Example 2. 
\begin{proof}[Proof of Example \ref{example:1}]
For notational simplicity we set $\rho =1$. Finding a compression algorithm for $\Bc_2^n(1)$ is equivalent to covering $\Bc_2^n(1)$ with $\ell_2$-balls of radius ${\d}$. Consider the following grid points for the interval $[-1,1]$:
\[
\Gc_1 = \left\{ -\left\lceil {\frac{\sqrt{n}}{{\d}}} \right\rceil {\frac{{\d}}{\sqrt{n}}}, \ldots , -{\frac{{\d}}{\sqrt{n}}}, 0 ,{\frac{{\d}}{\sqrt{n}}}, \ldots, \left\lceil {\frac{{\sqrt{n}}}{{\d}}} \right\rceil {\frac{{\d}}{\sqrt{n}}}\right\}. 
\]
It is straightforward to show that $\ell_2$-balls of radius ${\d}$ with centers on
\[
\Gc_n =  \underbrace{\Gc_1 \times \Gc_1 \times \ldots, \times \Gc_1}_n.
\]
covers the entire space $\Bc_2^n(1)$. Therefore, our compression scheme maps each vector to its closest codeword, i.e.,
\[
\Dc(\Ec(\xv)) = \arg \min _{\zv \in \Gc_n} \|\zv- \xv\|_2. 
\]
If the minimizer is not unique, the compression algorithm chooses one of the minimizers at random. The rate such compression algorithm achieves is equal to 
\begin{eqnarray*}
r(\d) &=& \log \left(2 \left\lceil \frac{\sqrt{n}}{{\d}} \right\rceil +1\right)^n \nonumber \\
   &\leq &  \log \left(2 \frac{\sqrt{n}}{{\d}} +3\right)^n \nonumber \\
   &\leq &  \log \left(5 \frac{\sqrt{n}}{{\d}}\right)^n \nonumber \\ 
&=& n\log(\sqrt{n}) + n \log\left(\frac{1}{{\d}} \right)+ n\log(5).  
\end{eqnarray*}

\end{proof}
\begin{proof}[Proof of Example \ref{example:2}]
Our encoding scheme is inspired by the previous example. The space of all $k$-sparse signals has ${n \choose k}$ hyperplanes. Once we specify the hyperplane $\mathcal{H}$, $\mathcal{H} \cap \Bc_2^n(1)$ is an $\ell_2$-ball of radius $1$ in $k$-dimensional subspace. Therefore, according to Example \ref{example:1} we require 
\[
\log \left(\frac{\sqrt{k}}{{\d}}\right)^k+ c k
\]
bits to code it with distortion smaller than ${\d}$ in a specified subspace. Therefore, overall we require $\log{n \choose k}$ for coding the subspace, and $\log \left(\frac{2\sqrt{k}}{{\d}}\right)^k$ for specifying the codeword on each hyperplane. This proves
\[
r(\d) = \log{n \choose k} + \log \left(\frac{\sqrt{k}}{{\d}}\right)^k + ck.
\]
\end{proof}
\subsection{Proofs of weak CS-applicability theorems}\label{proof:indcsapp}

\begin{proof}[Proof of Example \ref{example:3}]
Consider a compression code for $\Bc_2^n(1)$  with codebook $\Cc$ operating at rate $r$ and distortion $\d$. Each reconstruction codeword $\xvh\in\Cc$ can cover at most a ball of radius ${\d}$ in $\Bc_2^n(1)$. Hence, overall the $2^r$ codewords in $\Cc$, at most cover a volume equal to $2^r V_n {\d}^n$, where $V_n$denote the volume of $\Bc_2^n(1)$. Since the code has maximum distortion ${\d}$, these balls should cover the whole $\Bc_2^n(1)$. Hence, $2^r v_n {\d}^n \geq V_n$, or $r\geq n\log_2{1\over {\d}}$. This lower bound, combined with the upper bound that can be derived from Example \ref{example:1} proves that there exists a code with $\a$-dimension  equal to $n$.
\end{proof}
\begin{proof}[Proof of Theorem \ref{thm:CSP}]
Let  $\xvt_o=\Dc(\Ec(\xv_o))$, and $\xvh_o=\argmin_{\cv \in \Cc}  \|\yv_o- A\cv \|_2^2.$

Since $\xvh_o$ minimizes $ \|\yv_o- A\cv \|_2^2$ over all $\cv\in\Cc$, we have
\begin{align}
 \|\yv_o- A\xvh_o \|_2 &= \|A\xv_o- A\xvh_o \|_2\nonumber\\
			       & \leq \|A\xv_o- A\xvt_o \|_2\nonumber\\
			       & =\|\xv_o- \xvt_o\|_2\|\uv_1\|_2,
\end{align}
where
\[
\uv_1 \triangleq {A(\xv_o- \xvt_o) \over \|\xv_o- \xvt_o\|_2}.
\] 
Since the entries of $A$ are i.i.d.~Gaussian, $\uv_1$ is a vector of $d$ independent zero-mean Gaussian random variables  with variance $1$. For $\tau_1>0$, define event 
\[
\Ec_1\triangleq\{\|\uv_1\|_2^2< d(1+\tau_1)\}.
\] 
By Lemma \ref{lemma:chi},
\begin{align*}
\P(\Ec_1^c)=\P(\|\uv_1\|^2_2 \geq d(1+\tau_1))\leq {\rm e} ^{-\frac{d}{2}(\tau_1 - \log(1+ \tau_1))}.
\end{align*}
Since $\xvt_o$ is the reconstruction of $\xv_o$ using the  compression code $(\Ec,\Dc)$, it follows that  
$\|\xv_o- \xvt_o\|_2\leq {\d}$. Therefore, conditioned on $\Ec_1$, 
\begin{align}
 \|\yv_o- A\xvt_o \|_2 & \leq {\d}\sqrt{d(1+\tau_1)}.
\end{align}

To find a lower bound on $ \|\yv_o- A\xvh_o \|_2$, note that for a fixed $\cv\in\Cc$, 
\begin{align}
 \|\yv_o- A\cv \|_2&= \|A(\xv_o-\cv) \|_2\nonumber\\
&= \|\xv_o-\cv\|_2\|\uv_2 \|_2,\nonumber
\end{align}
where 
\[
\uv_2 \triangleq {A(\xv_o- \cv) \over \|\xv_o- \cv\|_2}.
\] 
Similar to $\uv_1$, $\uv_2$ is a $d$-dimensional distributed as $\Nc({\bf 0},I_d)$. Note that $\uv_2$ depends on $\cv$.
For $\tau_2\in(0,1)$, define event $\Ec_2$ as
\begin{align*}
\Ec_2\triangleq\{\forall\; \cv\in\Cc: \|\uv_2^2\|\geq d(1-\tau_2) \}.
 \end{align*}
By Lemma \ref{lemma:chi} and the union bound, it follows that
 \begin{align}
\P(\Ec_2^c)& \leq \sum_{\cv\in\Cc} \P \{\|\uv_2^2\|\leq d(1-\tau_2) \}\nonumber\\
&\leq2^{r}{\rm e} ^{\frac{d}{2}(\tau_2 + \log(1- \tau_2))}.
 \end{align}
Combining the two events, conditioned on that $\Ec_1$ and $\Ec_2$ both hold,
\begin{align}
\sqrt{d(1-\tau_2)}\|\xv_o-\xvh_o\|_2\leq {\d}\sqrt{d(1+\tau_1)}, 
\end{align}
or equivalently,
\begin{align}
\|\xv_o-\xvh_o\|_2\leq {\d}\sqrt{1+\tau_1\over 1-\tau_2}.
\end{align}
Finally, by the union bound,
\begin{align*}
\P(\Ec_1\cap\Ec_2)&=1-\P(\Ec_1^c\cup\Ec_2^c)\nonumber\\
&\geq1-2^{r}{\rm e} ^{\frac{d}{2}(\tau_2 + \log(1- \tau_2))}- {\rm e} ^{-\frac{d}{2}(\tau_1 - \log(1+ \tau_1))}.
\end{align*}
\end{proof}

\begin{proof}[Proof of Theorem \ref{thm:CSPdetnoiseind}]
The proof of this theorem is similar to the proof of Theorem \ref{thm:CSP}. There are a few differences that we highlight here. 
Let  $\xvt_o=\Dc(\Ec(\xv_o))$. Since $\xvh_o=\argmin_{\cv \in \Cc}  \|\yv_o- A\cv \|_2^2$, it follows that $\|\yv_o- A\xvh_o \|_2\leq \|\yv_o- A\xvt_o \|_2$. Therefore, by the triangle inequality, 
\begin{align}
\|A\xv_o- A\xvh_o \|_2-\|\zv_o\|_2\leq \|A\xv_o- A\xvt_o \|_2+\|\zv_o\|_2,\nonumber
\end{align}
or 
\begin{align}
\|A\xv_o- A\xvh_o \|_2 &\leq \|A\xv_o- A\xvt_o \|_2+2\|\zv_o\|_2 \nonumber \\
&\leq \|A\xv_o- A\xvt_o \|_2 + 2 \zeta.
\end{align}
In the rest of the proof we should provide an upper bound for $\|A\xv_o- A\xvt_o \|_2$ and a lower bound for  $\|A\xv_o- A\xvh_o \|_2$. 
They follow exactly as the proof of Theorem \ref{thm:CSP}.
\end{proof}

\vspace{.2cm}


\begin{proof}[Proof of Theorem \ref{thm:CSP-noisy}]
Let $\xvt_o=\Dc(\Ec(\xv_o))$. Since by assumption the code operates at distortion ${\d}$, we have $\|\xv_o-\xvt_o\|\leq {\d}$. On the other hand, since $\xvh_o$ is the solution of \eqref{eq:CSP},
\begin{align}
 \|\yv_o- A\xvh_o \|_2 &= \|A(\xv_o-\xvh_o)+\zv \|_2\nonumber\\
			       & \leq \|A(\xv_o-\xvt_o)+\zv \|_2.\label{eq:ub-thm2}
\end{align}
Expanding both sides of \eqref{eq:ub-thm2} and canceling the common terms, we obtain
\begin{align}
&\|A(\xv_o-\xvh_o)\|_2^2+2\zv^TA(\xv_o-\xvh_o)\nonumber\\
			       & \leq \|A(\xv_o-\xvt_o)\|_2^2+2\zv^TA(\xv_o-\xvt_o).\label{eq:up-cancel-terms}
\end{align}
Let
\[
\uv_1\triangleq {A(\xv_o-\xvt_o) \over \|\xv_o-\xvt_o\|_2},
\]
and
\[
\uv_2\triangleq {A(\xv_o-\xvh_o) \over \| A (\xv_o-\xvh_o) \|_2}.
\]
Using this definition,  along with triangle inequality and $a<|a|$, we rewrite  \eqref{eq:up-cancel-terms} as
\begin{align}
& \| A(\xv_o-\xvh_o) \|_2^2 -2 \| A(\xv_o-\xvh_o) \|_2|\zv^T\uv_2|\nonumber\\
			       & \leq \|\xv_o-\xvt_o\|_2^2 \|\uv_1\|_2^2+2\|\xv_o-\xvt_o\|_2|\zv^T\uv_1|.\label{eq:bounds-u}
\end{align}

For $\tau_1>0$, define events $\Ec_1$ as 
\[
\Ec_1\triangleq\Big\{\|A(\xv_o- \xvt_o)\|_2^2 \leq (1+\tau_1)d\|\xv_o- \xvt_o\|_2^2\Big\}.
\] 
Conditioned on $\Ec_1$ we can upper bound $ \|\uv_1\|_2$ by $\sqrt{d(1+\tau_1)}$.

In order to bound $|\zv^T\uv_1|$, we employ Lemma \ref{lemma:gaussian-vectors}. Given $\xv_o$, $\xvt_o$ and $\xvh_o$, both $\uv_1$  is i.i.d.~Gaussian vectors with mean zero and variance one, and is both independent of $\zv$. Therefore, by Lemma \ref{lemma:gaussian-vectors}, $\zv^T\uv_1$ is distributed as $\|\zv\|_2G_1$,  where $G_1$ is zero-mean variance-one Gaussian random variables independent  of $\|\zv\|_2$. For $\g_1>1$, define events $\Ec_2$ as:
\[
\Ec_2\triangleq \{ |\zv^T\uv_1| \leq \g_1\s \sqrt{d} \}.
\]
As argued above, 
\begin{align}
\P&(\Ec_2^c)=\P(\|\zv\|_2|G_1|\geq \g_1\s \sqrt{d} )\nonumber\\
=&\P(\|\zv\|_2|G_1|\geq \g_1\s \sqrt{d},  \|\zv\|_2\geq\s \sqrt{d(1+\tau_3) }\;)\nonumber\\
&+\P(\|\zv\|_2|G_1|\geq \g_1\s \sqrt{d},    \|\zv\|_2<\s\sqrt{ d(1+\tau_3) } \;)\nonumber\\
\leq&\P(\|\zv\|_2\geq\s \sqrt{d(1+\tau_3) }\;)
+\P(|G_1|\geq \g_1(1+\tau_3)^{-0.5}  \;)\nonumber\\
\leq&{\rm e}^{-\frac{d}{2}(\tau_3 - \log(1+ \tau_3))}+{\rm e}^{-\g_1^2/2(1+\tau_3)}.\label{eq:PE4}
\end{align}
where $\tau_3>0$, and the last line follows from Lemma \ref{lemma:chi}. 
Let
\[
\Ec_3\triangleq \{\forall \; \cv\in\Cc: |\zv^TA(\xv_o-\cv)| \leq  \g_2\s  \| A (\xv_o-\cv) \|_2\}.
\]
Note  that since $\|A(\xv_o-\cv)/\| A (\xv_o-\cv) \|_2 \|_2=1$,  for fixed $A(\xv_o-\cv)/\| A (\xv_o-\cv)\|_2$, $\zv^TA(\xv_o-\cv)/\| A (\xv_o-\cv)\|_2$ is distributed as $\Nc(0,\s^2)$. Hence, by the union bound,
\begin{align}
\P(\Ec_3^c)\leq&2^{r} {\rm e}^{-\g_2^2/2}.\label{eq:PE5}
\end{align}
Conditioned on $\Ec_1\cap\Ec_2\cap\Ec_3$,  \eqref{eq:bounds-u} yields
\begin{align}
&\| A(\xv_o-\xvh_o) \|_2^2-2 \g_2\sigma \|A(\xv_o-\xvh_o)\|_2\nonumber\\
&-(1+\tau_1)d{\d}^2-2\g_1\s {\d}\sqrt{d} \leq0.\label{eq:2nd-degree}
\end{align}
The quadratic equation $ax^2-2bx-c=0$, with $a,b,c>0$, has one positive and one negative root. Therefore,  we conclude that $\|A (\xv_o-\xvh_o)\|_2$ is  smaller than the positive root of \eqref{eq:2nd-degree}. That is, conditioned on $\Ec_1\cap\ldots\cap \Ec_3$,  $ \|A (\xv_o-\xvh_o)\|_2$ is upper bounded as
\begin{align}
 \|A(\xv_o -\xvh_o) \|_2& \leq \s\g_2
+ \sqrt{\s^2\g_2^2+(1+\tau_1)d{\d}^2+2\g_1\s {\d}\sqrt{d}}\label{eq:root-ub}
\end{align}

%

For $\tau_2\in (0,1)$, define event $\Ec_4$ as 
\begin{align*}
\Ec_4\triangleq\{\forall\; \cv\in\Cc: \|A(\xv_o-\cv)\|_2\geq \sqrt{d(1-\tau_2)}{\|\xv_o-\cv\|_2} \}.
 \end{align*}
Conditioned on $\Ec_1\cap\ldots\cap\Ec_4$, from \eqref{eq:root-ub}, we obtain
\begin{align}
 \|\xv_o -\xvh_o\|_2& \leq { \s\g_2+ \sqrt{\s^2\g_2^2+(1+\tau_1)d{\d}^2+2\g_1\s {\d}\sqrt {d}} \over \sqrt{d(1-\tau_2) } }.\label{eq:root-ub}
\end{align}


To set the free parameters, we analyze the probability of $\Ec_1\cap\ldots \cap\Ec_4$. By the union bound,
\[
\P(\Ec_1\cap\ldots\cap\Ec_4)\geq 1-\P(\Ec_1^c)-\ldots-\P(\Ec_4^c).
\]
To make sure that $\P(\Ec_1\cap\ldots\cap\Ec_4)$ is close to one for large values of $d$, it suffices to make $\P(\Ec_i^c)\to 0$, for $i=1,\ldots,4$. By Lemma \ref{lemma:chi},
\[
\P(\Ec_1^c)\leq {\rm e} ^{-\frac{d}{2}(\tau_1 - \log(1+ \tau_1))},
\] 
and by Lemma \ref{lemma:chi} and the union bound,
\begin{align}
\P(\Ec_4^c)\leq 2^{r}{\rm e} ^{\frac{d}{2}(\tau_2 + \log(1- \tau_2))}.\label{eq:E4-ub}
\end{align}
Upper bounds on $\P(\Ec_2^c)$ and $\P(\Ec_3^c)$ are given in \eqref{eq:PE4} and \eqref{eq:PE5}, respectively. 
Let $\tau_1=\tau_3=1$, $\tau_2=1-(\ex {\d})^{\frac{2}{\eta}(1+\epsilon') }$, 
\[
\g_1=\sqrt{4r \over \log_2{1\over {\rm e}{\d}}},
\]
and
\[
\g_2=\sqrt{2r}
\]
For $\tau_1=1$,
\[
\P(\Ec_1^c)\leq {\rm e} ^{-{0.15 \eta r \over \log_21/({\rm e}{\d})}}.
\] 
Similarly, for $\tau_3=1$, from \eqref{eq:PE4},
\begin{align}
\P(\Ec_2^c) \leq {\rm e} ^{-{0.15 \eta r \over \log_2 1/({\rm e}{\d})}}+{\rm e}^{-{r\over \log_2 1/({\rm e}{\d})}}.
\end{align}

For $\g_2=\sqrt{2r}$, since $\log 2-1<-0.3$, we obtain
\[
\P(\Ec_3^c)\leq \ex^{-0.3r}.
\]
Finally, for $\tau_2=1-(\ex {\d})^{\frac{2}{\eta}(1+\e')}$, as shown before in \eqref{eq:bound-p-details}, $r \log 2+0.5 d (\tau_2+\log(1-\tau_2))\leq -0.3\e r$.
Inserting the values of the parameters in \eqref{eq:root-ub}, we derive
\begin{eqnarray}
\lefteqn{\|\xv_o-\hat{\xv}_o\|_2} \nonumber \\
&\leq& \!\!\! {1\over (\ex {\d})^{(1+\e')/\eta}} \left({2\s \sqrt{\log_2(1/e{\d})} \over \sqrt{\eta}} \right) \nonumber \\
&+ &  \!\!\! {1\over (\ex {\d})^{(1+\e')/\eta}}  \left(\sqrt{{4\s^2 \log_2(1/e{\d}) \over \eta}+2{\d}^2+{4\s {\d}\over \sqrt{\eta}}}\right), \nonumber
\end{eqnarray}
which holds conditioned on $\Ec_1\cap\ldots\cap\Ec_4$.
\end{proof}
\subsection{Proofs of strong CS-applicability theorems}\label{proof:unifcsapp}
\begin{proof}[Proof of Theorem \ref{thm:finiteuniform}] \label{sec:finiteuniformproof}
 Let $\xv_o\in\Qc$, $\yv_o=A\xv_o$,  $\xvt_o=\Dc(\Ec(\xv_o))$, and $\xvh_o=\argmin_{\cv \in \Cc}  \|\yv_o- A\cv \|_2^2.$ As before, we have $\|\yv_o-A\hat{\xv}_o\|_2 \leq \|\yv_o- A \tilde{\xv}_o\|_2$. Hence,
\[
\| A \tilde{\xv}_o- A \hat{\xv}_o\|_2 - \|A \tilde{\xv}_o- A {\xv}_o \|_2  \leq \| A \xv_o - A \tilde{\xv}_o \|_2.
\]
Rearranging the terms proves that
\begin{equation}\label{eq:upperunif}
\| A \tilde{\xv}_o- A \hat{\xv}_o\|_2 \leq 2 \| A \xv_o - A \tilde{\xv}_o \|_2 \leq 2 \sigma_{\rm max}(A) {\d},
\end{equation}
where $\sigma_{\rm max}(A)$ is the maximum singular value of $A$. Define
\[
\Tc \triangleq \{\cv_1 -\cv_2 \ | \  (\cv_1,\cv_2)\in  \Cc_r\times \Cc_ \}.
\]
Note that $|\mathcal{T}| = 2^{2r}$. Given, $\tau\in(0,1)$, define  event $\Ec_1$ as
\begin{align}
\Ec_1  \triangleq\{ \forall \ \mathbf{h} \in \Tc \, ;  \,  \|A \mathbf{h}\|_2^2 > (1-\tau) d \|\mathbf{h}\|^2_2 \},\label{eq:E1}
\end{align}
and, for $t>0$, the event  $\Ec_2$ as
\begin{align}
\Ec_2 \triangleq \left\{\sigma_{\max}(A) - \sqrt{d} - \sqrt{n} < t \sqrt{d} \right\}.\label{eq:E2}
\end{align}
Conditioned on $\Ec_1$ and $\Ec_2$, \eqref{eq:upperunif} implies that
\begin{eqnarray}\label{eq:upperunif2}
\|\tilde{\xv}_o - \hat{\xv}_o \|_2 \sqrt{d(1-\tau) }  \leq 2 (\sqrt{n} + (t+1) \sqrt{d}) {\d}.
\end{eqnarray}
The last step is to find a lower bound for $\P(\Ec_1 \cap \Ec_2)$ or an upper bound for $\P(\Ec_1^c \cup \Ec_2^c)$. Note that $A \frac{ \mathbf{h}}{\|\mathbf{h}\|_2}$ is a vector of i.i.d.  $N(0,1)$ random variables. Using Lemma  \ref{lemma:chi} we obtain 
\begin{equation}\label{eq:probupp1}
\P( \|A\mathbf{h}\|_2^2 > (1 -\tau) d \|\mathbf{h}\|_2^2 ) \leq {\rm e} ^{\frac{d}{2}(\tau + \log(1- \tau))}.
\end{equation}
Employing union bound and \eqref{eq:probupp1} we have
\begin{equation}\label{eq:events1}
\P(\Ec_1^c) \leq  2^{2r}{\rm e} ^{\frac{d}{2}(\tau + \log(1- \tau))}.
\end{equation}
Finally, using the results on the  concentration of Lipschitz functions of a Gaussian random vector \cite{CaTa05}, we obtain
\begin{align}\label{eq:event2_nless}
\P\left(\Ec_2^{c}\right) &= \P\left(\sigma_{\max}(A) - \sqrt{d} - \sqrt{n} > t \sqrt{d} \right)\nonumber\\ 
&\leq {\rm e}^{-d t^2/2}.
\end{align}
This result is known as Davidson-Szarek theorem. Combining \eqref{eq:events1} and \eqref{eq:event2_nless} with \eqref{eq:upperunif2} finishes the proof.
\end{proof}

\begin{proof}[Proof of Theorem \ref{thm:csp_detnoise}]
The proof is very similar to the proof of Theorem \ref{thm:finiteuniform}.  Let $\xv_o\in\Qc$, $\yv_o=A\xv_o$,  $\xvt_o=\Dc(\Ec(\xv_o))$, and $\xvh_o=\argmin_{\cv \in \Cc}  \|\yv_o- A\cv \|_2^2.$ As before, we have $\|\yv_o-A\hat{\xv}_o\|_2 \leq \|\yv_o- A \tilde{x}_o\|_2$. Hence,
\[
\| A \tilde{\xv}_o- A \hat{\xv}_o\|_2 - \|A \tilde{\xv}_o- A {\xv}_o \|_2 - \|e\|_2  \leq \| A \xv_o - A \tilde{\xv}_o \|_2 + \|e\|_2.
\]
Rearranging the terms proves that
\begin{eqnarray}\label{eq:upperunifnoise}
\| A \tilde{\xv}_o- A \hat{\xv}_o\|_2 &\leq& 2 \| A \xv_o - A \tilde{\xv}_o \|_2 + 2 \|e\|_2 \nonumber \\
 &\leq& 2 \sigma_{\rm max}(A) {\d} + 2 \zeta,
\end{eqnarray}
Defining $\Ec_1$ and $\Ec_2$ as in the proof of Theorem \ref{thm:finiteuniform} and following the same approach proves:
\begin{equation}\label{eq:events1noise}
\P(\Ec_1^c) \leq  2^{2r}{\rm e} ^{\frac{d}{2}(\tau + \log(1- \tau))},
\end{equation}
and
\begin{align}\label{eq:event2_nlessnoise}
\P\left(\Ec_2^{c}\right) &= \P\left(\sigma_{\max}(A) - \sqrt{d} - \sqrt{n} > t \sqrt{d} \right)\nonumber\\ 
&\leq {\rm e}^{-d t^2/2}.
\end{align}
Combining \eqref{eq:events1noise} and \eqref{eq:event2_nlessnoise} with \eqref{eq:upperunifnoise} finishes the proof.
\end{proof}


\begin{proof}[Proof of Theorem \ref{thm:stocnoiseunif}]
Let $\tilde{\xv_o} = \mathcal{D} (\Ec (\xv_o))$. Since $\xvh_o = \arg\min_{\mathbf{c} \in \mathcal{C}} \|\yv_o - A \mathbf{c} \|_2^2$, $\|\yv_o - A \xvh_o \|_2 \leq \|\yv_o - A \tilde{\xv}_o \|_2$. Therefore,
\begin{align}
\| \yv_o - A\tilde{\xv}_o -A (\xvh_o-\tilde{\xv}_o)\|_2^2 \leq  \|\yv_o - A \tilde{\xv}_o \|^2_2
\end{align}
and consequently 
\begin{align}
&\|A (\xvh_o-\tilde{\xv}_o)\|_2^2 -2(\yv_o - A\tilde{\xv}_o)^TA (\xvh_o-\tilde{\xv}_o)=\nonumber\\
&\|A (\xvh_o-\tilde{\xv}_o)\|_2^2 -2(A(\xv_o-\tilde{\xv}_o)-\zv)^TA (\xvh_o-\tilde{\xv}_o)\leq0.\label{eq:step1-thm9}
\end{align}
Let 
\[
\mathbf{u} \triangleq \frac{A( \xvt_o-\xvh_o)}{\| A( \xvt_o-\xvh_o)\|_2}.
\]
Using the Cauchy-Schwarz inequality, from \eqref{eq:step1-thm9} we obtain
\begin{align}
\|&A \xvt_o - A \xvh_o\|_2^2 - 2 \| A (\xvt_o - \xvh_o)\|_2 \|  A(\xv_o - \xvt_o) \|_2 \nonumber\\
& -2 |\zv^T \mathbf{u} | \|A(\xvt_o - \xvh_o)\|_2 \leq 0
\end{align}
or
\begin{align}
\|&A \xvt_o - A \xvh_o\|_2 \leq 2 \|  A(\xv_o - \xvt_o) \|_2+2 |\zv^T \mathbf{u} |.\label{eq:uppernoiseunif1}
\end{align}
Define 
\[
\Ec_1  \triangleq \{ \sigma_{\max}(A) - \sqrt{d} - \sqrt{n} < t \sqrt{d} \},
\]
and 
\[
\Ec_2 \triangleq \{ \forall \cv_1, \cv_2 \in \Cc \ : \  |\zv^T A(\cv_1-\cv_2) | \leq \gamma \s \|A(\cv_1-\cv_2)\|_2\}.
\]
Note that conditioned on $\Ec_2$, $|\zv^T\uv|\leq \g_1\s$. Conditioned on $\Ec_1$ and $\Ec_2$, we can simplify \eqref{eq:uppernoiseunif1} and obtain
\begin{align}
\|A(\xvh_o - \xvt_o)\|_2 \leq 2 (\sqrt{n} + (t+1)\sqrt{d}) {\d} + 2 \gamma \sigma.\label{eq:uppernoiseunif2}
\end{align}
For $\tau\in(0,1)$, let
\[
\Ec_3 \triangleq \{ \forall \cv_1, \cv_2 \in \Cc: \  \| A(\cv_1 - \cv_2) \|_2 \geq \sqrt{(1- \tau)d} \|\cv_1-\cv_2\|_2   \}.
\]
Conditioned on $\Ec_3$ we can simplify \eqref{eq:uppernoiseunif2} and obtain
\[
\| \xvh_o- \xvt_o\|_2 \leq \frac{2(\sqrt{n} + (t+1) \sqrt{d}) {\d} + 2 \gamma \sigma }{ \sqrt{(1- \tau)d}}.
\]
Therefore, 
\begin{align}
\| \xvh_o- \xv_o\|_2& \leq \| \xvh_o- \xvt_o \| + \|\xvt_o- \xv_o \|_2\nonumber\\
&\leq  \frac{2(\sqrt{n} + (t+1) \sqrt{d}) {\d} + 2 \g \sigma }{(1- \tau) \sqrt{d}} + {\d}.\label{eq:finalstocnoiseu}
\end{align}
To finish the proof we need to find an upper bound on  $\P(\Ec_1^c \cup \Ec_2^c \cup \Ec_3^c  )$.
As mentioned before in the proof of Theorem \ref{thm:finiteuniform} we have
\begin{equation}
\P(\Ec_1^c) \leq {\rm e}^{-\frac{dt^2}{2}}. \label{eq:prob1stocnoiseu}
\end{equation}
To obtain an upper bound for $\P(\Ec_2^2)$, note that given $\cv_1,\cv_2\in \mathcal{C}$, $\zv^T(\cv_1-\cv_2)/\|\cv_1-\cv_2\| \sim \Nc(0,\sigma^2)$. Hence, by the union bound, since $|\Cc|\leq 2^r$, we obtain
\begin{align}
\P(\Ec_2^c) &\leq 2^{2r}  {\rm e}^{-\frac{\g^2}{2}}\label{eq:prob2stocnoiseu}
\end{align}
Finally, following a similar argument and Lemma \ref{lem:chisq} yield
\begin{align}
\P(\Ec_3^c) &\leq  2^{2r} {\rm e}^{\frac{d}{2} (\tau+ \log(1 - \tau))}.\label{eq:prob3stocnoiseu}
\end{align}
Combining \eqref{eq:prob1stocnoiseu}, \eqref{eq:prob2stocnoiseu}, \eqref{eq:prob3stocnoiseu} and \eqref{eq:finalstocnoiseu} completes the proof.
\end{proof}



\subsection{Proofs of analog CS} \label{proof:analogcs}
\begin{lemma}\label{lem:chisq}
Let $f \in L_2([0,1])$ and consider
\[
\mathcal{A}(f) = \left[\int_t f dW_1(t), \ldots, \int_t fdW_d(t)\right],
\]
where $W_1, W_2, \ldots, W_d$ are independent Brownian motions. Then $\frac{\|\Ac(f)\|_2^2}{\|f\|_2^2}$ is a $\chi^2$ random variable with $d$ degrees of freedom.
\end{lemma}
\begin{proof}
According to Theorem \ref{thm:gaussian_wiener} and independence of $W_i$s, the elements of $\Ac(f)$ are iid $N(0,\|f\|_2^2)$. Therefore $\Ac(f)/\|f\|_2$ is iid $N(0,1)$ vector, that proves \[
\frac{\|\Ac(f)\|_2^2}{\|f\|_2^2} \sim \chi^2(d).
\]
\end{proof}

\begin{proof}[Proof of Theorem \ref{thm:CSPInf}]
Given Lemmas \ref{lem:chisq} and Theorem \ref{thm:gaussian_wiener}  the proof is essentially the same as the proof of Theorem \ref{thm:CSP}.  Therefore, we briefly mention the main steps. Let $\tilde{f}_o =\Dc(\Ec(f_o))$ and $\hat{f}_o$ denote the reconstruction of CSP. Since $\hat{f}_o$ is the solution of $CSP$, we have
\[
\|\Ac(\hat{f}_o )- \Ac(f_o)\|_2^2 \leq \|\Ac(\tilde{f}_o )- \Ac(f_o)\|_2^2 
\]
According to Lemma \ref{lem:chisq} both $\|\Ac(\tilde{f}_o)- \Ac(f_o)\|_2^2/ \|\tilde{f}_o- f_o\|_2^2$ and $\|\Ac(\hat{f}_o)- \Ac(f_o)\|_2^2/ \|\hat{f}_o-f_o\|_2^2$ are $\chi^2$ random variables with $d$ degrees of freedom. Therefore it is straightforward to employ Lemma \ref{lemma:chi} and confirm
\[
\P(\|\Ac(\tilde{f}_o)- \Ac(f_o)\|_2  \geq {\d}\sqrt{d(1+ \tau)} ) \leq  {\rm e}^{-\frac{d}{2}(\tau - \log(1+\tau))},
\]
and also for every $\hat{f} \in \Cc_r$
\[
\|\Ac(\hat{f}_o)- \Ac(f_o)\|_2 \geq \|f_o- \hat{f}_o\|_2 \sqrt{d(1-\tau)}. 
\]
with probability $1-2^{r} {\rm e}^{d/2(\tau+ \log(1-\tau))}$.
Combining these results completes the proof.
\end{proof}

\section{Conclusion}\label{sec:conclusion}
In this paper, we studied the problem of employing a family of compression algorithms for compressed sensing, \ie recovering structured signals from their  undersampled set of random linear measurements. Addressing this problem  enables CS schemes to exploit complicated structures  integrated in compression algorithms.  We  proposed compressible signal pursuit (CSP) algorithm that outputs the codeword that  best matches the measurements.  We  proved that employing a family of compression algorithms whose  rate-distortion function satisfies $\lim \sup_{\d \rightarrow 0} {r(\d)}/{\log_2(1/\d)} \leq \alpha$, with $\alpha$ smaller than the ambient dimension, with high probability, CSP recovers signals from $2 \alpha$ measurements. We have also shown that this bound is sharp and the signal cannot be recovered if we have fewer measurements. CSP is also robust to measurement noise. Finally, CSP is also applicable to infinite-dimensional signal classes. CSP is still computationally demanding and requires approximation or simplification for practical applications. This important direction is left for  future research.

\bibliographystyle{unsrt}
\bibliography{myrefs}

\end{document}

%% file: defns.tex
\usepackage{xspace}
\usepackage{bbm}

\newcommand{\Ac}{\mathcal{A}}
\newcommand{\Bc}{\mathcal{B}}
\newcommand{\Cc}{\mathcal{C}}
\newcommand{\Dc}{\mathcal{D}}
\newcommand{\Ec}{\mathcal{E}}
\newcommand{\Fc}{\mathcal{F}}
\newcommand{\Gc}{\mathcal{G}}

\newcommand{\Nc}{\mathcal{N}}

\newcommand{\Qc}{\mathcal{Q}}

\newcommand{\Sc}{\mathcal{S}}
\newcommand{\Tc}{\mathcal{T}}

\newcommand{\Xv}{{\bf X}}
\newcommand{\Yv}{{\bf Y}}

\newcommand{\xv}{{\bf x}}
\newcommand{\yv}{{\bf y}}
\newcommand{\zv}{{\bf z}}
\newcommand{\uv}{{\bf u}}

\newcommand{\fh}{{\hat{f}}}

\def\a{\alpha}
\def\b{\beta}
\def\g{\gamma}
\def\d{\delta}
\def\e{\epsilon}

\let\P\relax
\DeclareMathOperator\P{P}

\newcommand\ie{i.e.,\xspace}
\def\textiid{i.i.d.\@\xspace}
\newcommand\iid{\ifmmode\text{ i.i.d. } \else \textiid \fi}

\newcommand{\ind}{\mathbbmss{1}}
